\documentclass[journal]{IEEEtran}
\usepackage{amsfonts}
\usepackage{amsmath}
\usepackage{amssymb}
\usepackage{graphicx}
\usepackage{amsthm}
\usepackage[font={footnotesize},labelfont={bf}]{caption}
\providecommand{\U}[1]{\protect\rule{.1in}{.1in}}

\newtheoremstyle{example}{\topsep}{\topsep}
{}
{}
{\bfseries}
{}
{  }
{\thmname{#1}\thmnumber{ #2}. \thmnote{ (#3)}}
\newtheorem{theorem}{Theorem}

\newtheorem{conjecture}[theorem]{Conjecture}
\newtheorem{corollary}[theorem]{Corollary}

\newtheorem{lemma}[theorem]{Lemma}

\theoremstyle{example}
\newtheorem{example}[theorem]{Example}

\newcommand{\nc}{\newcommand}
\nc{\rnc}{\renewcommand}
\nc{\beq}{\begin{equation}}
\nc{\eeq}{{\end{equation}}}
\nc{\beqa}{\begin{eqnarray}}
\nc{\eeqa}{\end{eqnarray}}
\nc{\lbar}[1]{\overline{#1}}
\nc{\bra}[1]{\langle#1|}
\nc{\ket}[1]{|#1\rangle}
\nc{\ketbra}[2]{|#1\rangle\!\langle#2|}
\nc{\braket}[2]{\langle#1|#2\rangle}
\nc{\proj}[1]{| #1\rangle\!\langle #1 |}
\nc{\avg}[1]{\langle#1\rangle}

\nc{\smfrac}[2]{\mbox{$\frac{#1}{#2}$}}
\nc{\tr}{\operatorname{tr}}
\nc{\tracedist}[1]{\Delta_{}\!\left( #1 \right)}
\nc{\fid}[1]{F\!\left( #1 \right)}

\nc{\ox}{\otimes}
\nc{\dg}{\dagger}
\nc{\dn}{\downarrow}
\nc{\cA}{{\cal A}}
\nc{\cB}{{\cal B}}
\nc{\cC}{{\cal C}}
\nc{\cD}{{\cal D}}
\nc{\cE}{{\mathcal E}}
\nc{\cF}{{\cal F}}
\nc{\cG}{{\cal G}}
\nc{\cH}{{\cal H}}
\nc{\cI}{{\cal I}}
\nc{\cJ}{{\cal J}}
\nc{\cK}{{\cal K}}
\nc{\cL}{{\cal L}}
\nc{\cM}{{\cal M}}
\nc{\cN}{{\cal N}}
\nc{\cO}{{\cal O}}
\nc{\cP}{{\cal P}}
\nc{\cR}{{\cal R}}
\nc{\cS}{{\cal S}}
\nc{\cT}{{\cal T}}
\nc{\cU}{{\cal U}}
\nc{\cX}{{\cal X}}
\nc{\cZ}{{\cal Z}}

\nc{\entI}{{\bf I}}
\nc{\entIarrow}{{\bf I}^{\leftarrow}}
\nc{\entH}{{\bf H}}
\nc{\entS}{{\bf S}}
\nc{\entHmin}{H_{\min}}

\nc{\entF}{{\bf E}_f}

\nc{\isom}{\simeq}

\nc{\rank}{\operatorname{rank}}
\nc{\rar}{\rightarrow}
\nc{\lrar}{\longrightarrow}
\nc{\polylog}{\operatorname{polylog}}
\nc{\poly}{\operatorname{poly}}
\nc{\1}{{\openone}}

\nc{\weight}{\textbf{w}}
\nc{\hamdist}{d_{H}}

\def\e{\epsilon}

\def\U{\Upsilon}

\nc{\Sp}{{{\mathbb S}}}
\nc{\RR}{{{\mathbb R}}}
\nc{\CC}{{{\mathbb C}}}
\nc{\FF}{{{\mathbb F}}}
\nc{\NN}{{{\mathbb N}}}
\nc{\ZZ}{{{\mathbb Z}}}
\nc{\PP}{{{\mathbb P}}}
\nc{\QQ}{{{\mathbb Q}}}
\nc{\UU}{{{\mathbb U}}}
\nc{\OO}{{{\mathbb O}}}
\nc{\EE}{{{\mathbb E}}}
\nc{\id}{{\operatorname{id}}}

\nc{\qubitchannel}{\id_2}
\nc{\bitchannel}{\overline{\id}_2}

\nc{\be}{\begin{equation}}
\nc{\ee}{{\end{equation}}}
\nc{\bea}{\begin{eqnarray}}
\nc{\eea}{\end{eqnarray}}
\nc{\<}{\langle}
\rnc{\>}{\rangle}
\nc{\Hom}[2]{\mbox{Hom}(\CC^{#1},\CC^{#2})}
\nc{\rU}{\mbox{U}}

\nc{\ob}[1]{#1}

\nc{\red}[1]{{\color{red} #1}}
\nc{\rem}[1]{{\color{gray} #1}}

\usepackage[utf8]{inputenc}

\begin{document}

\title{Classical communication over a quantum interference channel}

\author{Omar Fawzi, Patrick Hayden, Ivan Savov, Pranab Sen, and Mark M. Wilde\thanks{Omar Fawzi, Patrick Hayden, Ivan Savov, and Mark M. Wilde are with the School of Computer Science, McGill University, Montr\'{e}al, Qu\'{e}bec, Canada H3A 2A7. Pranab Sen is with the School of Technology and Computer Science, Tata Institute of Fundamental Research, Mumbai, India. Patrick Hayden was a visitor with the Perimeter Institute for Theoretical Physics, Waterloo, Ontario, Canada while conducting this research, and Pranab Sen was a visitor with the School of Computer Science, McGill University.

This paper was presented in part at the Forty-Ninth Annual Allerton Conference, Sept.~28-30, 2011 and the
2012 Quantum Information Processing Conference, Dec.~12-16, 2011.}}
\maketitle

\begin{abstract}
Calculating the capacity of interference channels is a notorious open
problem in classical information theory. Such channels have two senders and
two receivers, and each sender would like to communicate with a partner
receiver. The capacity of such channels is known exactly in the settings of
\textquotedblleft very strong\textquotedblright\ and \textquotedblleft
strong\textquotedblright\ interference, while the Han-Kobayashi coding strategy
gives the best known achievable rate region in the general case.

Here, we introduce and study the quantum interference channel, a natural
generalization of the interference channel to the setting of quantum
information theory. We restrict ourselves for the most part to channels with
two classical inputs and two quantum outputs in order to simplify the
presentation of our results (though generalizations of our results to channels
with quantum inputs are straightforward). We are able to determine the exact
classical capacity of this channel in the settings of \textquotedblleft very
strong\textquotedblright\ and ``strong'' interference, by exploiting Winter's
successive decoding strategy and a novel two-sender
quantum simultaneous decoder, respectively. We provide a proof that a
Han-Kobayashi strategy is achievable with Holevo information rates, up to a
conjecture regarding the existence of a three-sender quantum simultaneous decoder.
This conjecture holds for a special class of quantum multiple access
channels with average output states that commute, and we discuss some
other variations of the conjecture that hold. Finally, we detail a connection
between the quantum interference channel and prior work on the capacity of
bipartite unitary gates.

\end{abstract}


\section{Introduction}

Classical information theory came as a surprise to the communication engineers
of the 1940s and '50s~\cite{P73,V98}. It was astonishing that two-terminal
noisy communication channels generally have a non-zero capacity at which two
parties can communicate error-free in the asymptotic limit of many channel
uses, and furthermore, that the computation of this capacity is a
straightforward convex optimization problem~\cite{bell1948shannon}---many
consider the achievements of Shannon to be among the great scientific
accomplishments of the last century. Soon after this accomplishment, Shannon
laid the foundations for multi-user information theory, and he claimed that a
three-terminal communication channel with two senders and one receiver also
has a simple, elegant solution \cite{S61,V98}.\ Some time later, Liao and
Ahlswede provided a formal proof of the capacity of this multiple access
channel without any knowledge of Shannon's unpublished solution~\cite{L72,A74}%
. The beauty of information theory in these two settings is that it offers
elementary solutions to problems that, at the outset, seem to be
extraordinarily difficult to solve.

The situation for more general communication scenarios in multi-user
information theory is not as simple and elegant as it is for single-sender,
single-receiver channels and multiple access channels~\cite{el2010lecture}.
For example, the capacity of the interference channel is one of the notorious
open problems in classical information theory~\cite{K06}. The interference
channel refers to the setting in which a noisy communication channel connects
two senders to two receivers, and each sender's goal is to communicate with a
partner receiver. Each sender's transmission can interfere with the other's,
and this is one reason (among many) that the problem is difficult to solve in
the general case. This channel arises naturally in the context of data
transmission over interfering wireless links or digital subscriber
lines~\cite{K06}. Shannon himself introduced the problem and attempted to
solve it~\cite{S61}, but it is the later work of others that would
provide ongoing improvements to the inner and outer bounds for the capacity of
the interference
channel~\cite{carleial1975case,Sato77,Sato1978,S78,sato1981capacity,HK81,K04}.

Carleial offered the first surprising result for the interference
channel~\cite{carleial1975case}, by demonstrating that each sender can
achieve the same rates of communication as if there is no interference at all if the interference from the
other sender's transmission is \textquotedblleft very
strong.\textquotedblright\ Carleial's solution is to have each receiver decode
the other sender's message first and follow by decoding the partner sender's
message, rather than each receiver simply treating the other sender's
transmission as noise. Thus, Carleial's strategy demonstrates that we can
achieve improved communication rates by taking advantage of interference
rather than treating it as an obstacle. Sato then gave a full characterization
of the capacity of the Gaussian interference channel in the setting of
\textquotedblleft strong\textquotedblright%
\ interference~\cite{sato1981capacity}, by appealing to an earlier result of
Ahlswede regarding the capacity of a compound multiple access
channel~\cite{A74}. Han and Kobayashi independently found Sato's result,
and they built on these insights and applied them
to the most general setting (not necessarily \textquotedblleft
strong\textquotedblright\ or \textquotedblleft very strong\textquotedblright%
\ interference) %
by allowing for each decoder to partially decode the other sender's message 
and use this information to better decode the message intended for them~\cite{HK81}. 
The resulting achievable rate region
is known as the Han-Kobayashi rate region, and it is currently the best known
inner bound on the capacity of the interference channel.\footnote{Chong,
Motani, and Garg subsequently proposed another achievable rate region
originally thought to improve the Han-Kobayashi rate region \cite{CMG06}, but
later work demonstrated that the Chong-Motani-Garg achievable rate region is
equivalent to the Han-Kobayashi region \cite{CMGE08,kobayashi2007further}.}

The model of the interference channel as stated in the above works is an
important practical model for data transmission in a noisy two-input,
two-output network, but it ignores a fundamental aspect of the physical
systems employed to transmit this data. At bottom, these physical systems
operate according to the laws of quantum mechanics~\cite{book2000mikeandike},
and ultimately, at some level, these laws govern how noise can affect these
systems. Now, for many systems (macroscopic ones in particular), these laws
are not necessary to describe the dynamics of encoding, transmission, and
decoding, and one could argue in this case that there is not any benefit to
recasting information theory as a \textit{quantum} information theory because
it would only add a layer of complexity to the theory. However, there are
examples of natural physical systems, such as fiber optic cables or free space
channels, for which quantum information theory offers a boost in capacity if
the coding scheme makes clever use of quantum
mechanics~\cite{PhysRevLett.92.027902}. Thus, it is important to determine the
information capacities of quantum channels, given that the physical carriers
of information are quantum and quantum effects often give a boost in capacity.
In Ref.~\cite{PhysRevLett.92.027902}, it is shown that a receiver making use of a collective
measurement operating on all of the channel outputs has an improvement in
performance over a receiver decoding with single-channel-output measurements.
Additionally, there are existential arguments for examples of channels in
which entanglement at the encoder can improve performance, leading to
superadditive effects that simply cannot occur in classical information
theory~\cite{H09}.

The quantum-mechanical approach to information theory has shed a new light on
the very nature of information, and researchers have made much progress on
this front in the past few decades \cite{book2000mikeandike}. Perhaps the most
fundamental problem in quantum information theory is the task of transmitting
bits over a quantum channel. Holevo and Schumacher-Westmoreland (HSW)
offered independent proofs that the Holevo information, one generalization of
Shannon's mutual information, is an achievable rate for classical data transmission over
a quantum channel \cite{ieee1998holevo,PhysRevA.56.131}. Many researchers
thought for some time that the Holevo information of a quantum channel would
be equal to its classical capacity, but recent work has demonstrated that the
answer to the most fundamental question of the classical capacity of a quantum
channel remains wide open in the general case \cite{H09,HW08}.

Soon after the HSW\ result, quantum information theorists began exploring other avenues, one of
which is multi-user quantum information theory. Winter proved that the
capacity region of a quantum multiple access channel is a natural
generalization of the classical solution, in which we can replace Shannon
information rates with Holevo information rates~\cite{winter2001capacity}. It
was not obvious at the outset that this solution would be possible---after
all, any retrieval of data from a quantum system inevitably disturbs the state
of the system, suggesting that successive decoding strategies employed in the
classical case might not work for quantum systems~\cite{book1991cover}. But
Winter overcame this obstacle by realizing that a so-called \textquotedblleft
gentle\textquotedblright\ or \textquotedblleft tender\textquotedblright%
\ measurement, a measurement with an outcome that succeeds with high
probability, effectively causes no disturbance to the state in the asymptotic
limit of many channel uses. Later, Yard \textit{et al}.~considered various
capacities of a quantum broadcast channel~\cite{YHD2006}, and they found
results that are natural generalizations of results from classical multi-user
information theory~\cite{B73,el2010lecture}. In parallel with these
developments, researchers have considered many generalizations of the above
settings, depending on the form of the transmitted information
\cite{PhysRevA.55.1613,capacity2002shor,ieee2005dev,qcap2008first,nature2005horodecki,Horodecki:2007:107}%
, whether assisting resources are available
\cite{ieee2002bennett,SSW08,itit2008hsieh,DHL10}, or whether the sender and
receiver would like to trade off different resources against each other
\cite{cmp2005dev,DHW08,HW08a}.

\section{Summary of Results}

In this paper, we introduce the quantum interference channel, a natural
generalization of the interference channel to the quantum domain. We at first
restrict our discussion to a particular \textit{ccqq} quantum interference
channel, which has two classical inputs and two quantum outputs. 
This restriction simplifies the presentation,
and a straightforward extension of our results leads to results for a general
quantum interference channel with quantum inputs and quantum outputs. We summarize
our main results below:

\begin{itemize}

\item Our first
contribution is an exact characterization of the capacity region of a
\textit{ccqq} quantum interference channel with \textquotedblleft very
strong\textquotedblright\ interference---the result here is a straightforward
generalization of Carleial's result from Ref.~\cite{carleial1975case}.

\item Our second contribution
is a different exact characterization of the capacity of a \textit{ccqq} channel that
exhibits ``strong interference.'' This result employs a novel quantum simultaneous decoder
for quantum multiple access channels with two classical inputs and one quantum output.

\item Our next contribution is a quantization of the Han-Kobayashi achievable rate
region, up to a conjecture regarding the existence of a quantum simultaneous
decoder for quantum multiple access channels with three classical inputs and one quantum output.
We prove that a three-sender quantum
simultaneous decoder exists in the special case where the induced channel to
each receiver has average output states that commute, but we have not been able to prove the
existence of such a decoder in the general case (neither is it clear how to leverage the proof
of the two-sender simultaneous decoder). 
We prove that a certain rate region described in terms of min-entropies~\cite{R60,R05} is achievable for the general non-commuting case,
and our suspicion is that a proof for the most
general case should exist and will bear similarities to these proofs. The
existence of such a simultaneous decoder immediately implies that the senders
and receivers can achieve the rates  on the Han-Kobayashi inner bound. This conjecture
is also closely related to the the ``multiparty typicality'' conjecture formulated in \cite{D11}.

\item 
We also describe an achievable rate region for the quantum interference channel based on 
a successive decoding and rate splitting strategy~\cite{sasoglu2008successive}.

\item
We supply an outer bound on the capacity of the quantum
interference channel, similar to Sato's outer bound from Ref.~\cite{Sato1978}.

\item
Finally, we discuss the connection between prior work on the capacity of
unitary gates \cite{BennettHLS03,HarrowLeung05,HarrowLeung08,HarrowShor10}%
\ and the capacity of the quantum interference channel. The quantum
interference channel that we consider in this last contribution is an
isometry, in which the two inputs and two outputs are quantum and the channel
acts as a noiseless evolution from the senders to the receivers.

\end{itemize}

We structure this paper as follows. We first introduce the notation used in the
rest of the paper. We then detail the general information processing task that
two senders and two receivers are trying to accomplish using the quantum
interference channel. Section~\ref{sec:MAC} discusses the connection
between the multiple access channel and the interference channel, and we
prove the existence of a quantum simultaneous decoder for the multiple
access channel with two classical inputs and one quantum output. This section also
 states a conjecture regarding the existence of a quantum simultaneous decoder
with three classical inputs and one quantum output, and we prove that it
exists for a special case. We also discuss an achievable rate
region in terms of min-entropies, and we remark briefly on many avenues that
we pursued in an attempt to prove this conjecture. Section~\ref{sec:QIC}%
\ presents our results regarding the quantum interference channel. We first
determine the capacity of the quantum interference channel if the channel has
\textquotedblleft very strong\textquotedblright\ interference and follow with the capacity
when the channel exhibits ``strong'' interference. We next show
how to achieve the Han-Kobayashi inner bound, by exploiting the conjecture
regarding the existence of a three-sender quantum simultaneous decoder.
We then present a set of achievable rates obtained using successive decoding
and rate splitting. This section ends with an outer bound on the capacity of the
quantum interference channel. Section~\ref{sec:gate-capacities}\ presents our
final contribution regarding the connection to unitary gate capacities, and
the conclusion summarizes our findings and states open lines of pursuit for the quantum interference channel.

\section{Notation}

We denote quantum systems as $A$, $B$, and $C$ and their corresponding Hilbert
spaces as $\mathcal{H}^{A}$, $\mathcal{H}^{B}$, and $\mathcal{H}^{C}$ with
respective dimensions $d_{A}$, $d_{B}$, and $d_{C}$. We denote pure states of
the system $A$ with a \emph{ket} $\left\vert \phi\right\rangle ^{A}$ and the
corresponding density operator as $\phi^{A}=\left\vert \phi\right\rangle\!
\left\langle \phi\right\vert ^{A}$. All kets that are quantum states have unit
norm, and all density operators are positive semi-definite with unit trace. We
model our lack of access to a quantum system with the partial trace operation.
That is, given a two-qubit state $\rho^{AB}$ shared between Alice and Bob, we
can describe Alice's state with the reduced density operator:%
\[
\rho^{A}=\text{Tr}_{B}\left\{  \rho^{AB}\right\}  ,
\]
where Tr$_{B}$ denotes a partial trace over Bob's system. Let%
\[
H(A)_{\rho}\equiv-\text{Tr}\left\{  \rho^{A}\log\rho^{A}\right\}
\]
be the von Neumann entropy of the state $\rho^{A}$. For a state
$\sigma^{ABC}$, we define the quantum conditional entropy%
\[
H(A|B)_{\sigma}\equiv H(AB)_{\sigma}-H(B)_{\sigma},
\]
the quantum mutual information%
\[
I(A;B)_{\sigma}\equiv H(A)_{\sigma}+H(B)_{\sigma}-H(AB)_{\sigma},
\]
and the conditional quantum mutual information%
\[
I(A;B|C)_{\sigma}\equiv H(A|C)_{\sigma}+H(B|C)_{\sigma}-H(AB|C)_{\sigma}.
\]
Quantum operations are completely positive trace-preserving (CPTP) maps
$\mathcal{N}^{A^{\prime}\rightarrow B}$, which accept input states in
$A^{\prime}$ and output states in $B$. In order to describe the
\textquotedblleft distance\textquotedblright\ between two quantum states, we
use the notion of \emph{trace distance}. %
The trace distance between states $\sigma$ and $\rho$ is%
\[
\Vert\sigma-\rho\Vert_{1}=\mathrm{Tr}\left\vert \sigma-\rho\right\vert ,
\]
where $|X|=\sqrt{X^{\dagger}X}$. 
Two states that are similar have trace distance close to zero,  
whereas states that are perfectly distinguishable have trace distance equal to two.
Throughout this paper, logarithms and
exponents are taken base two unless otherwise specified.
Appendix~\ref{sec:typ-review} reviews several important properties of typical
sequences and typical subspaces.

\section{The Information Processing Task}

We first discuss the information processing task that two senders and two
receivers are trying to accomplish with the quantum interference channel. We
assume that they have access to many independent uses of a particular type of
channel with two classical inputs and two quantum outputs. A \textit{ccqq}
quantum interference channel is the following map:%
\begin{equation}
x,y\rightarrow\rho_{x,y}^{B_{1}B_{2}}\text{,} \label{eq:ccqq-int-channel}%
\end{equation}
where the inputs $x$ and $y$ produce a density operator $\rho_{x,y}%
^{B_{1}B_{2}}$ that exists on quantum systems $B_{1}$ and $B_{2}$. Receiver~1
has access to system $B_{1}$, and Receiver~2 has access to system $B_{2}$. An
$\left(  n,R_{1}-\delta,R_{2}-\delta,\epsilon\right)  $ quantum interference
channel code consists of three steps:\ encoding, transmission, and decoding.

\textbf{Encoding}. Sender~1 chooses a message $l$\ from a message set
$\mathcal{L}=\left\{  1,2, \ldots, |\mathcal{L}|\right\}  $ where
$|\mathcal{L}|=2^{n(R_{1}-\delta)}$, and Sender~2 similarly chooses a message
$m$ from a message set $\mathcal{M}=\left\{  1,2,\ldots,|\mathcal{M}|\right\}
$ where $|\mathcal{M}|=2^{n(R_{2}-\delta)}$, where $\delta$ is some
arbitrarily small positive number. Senders~1 and 2 then encode their messages
as codewords of the following form:%
\begin{align*}
x^{n}\!\left(  l\right)   &  \equiv x_{1}\!\left(  l\right)  \ x_{2}\!\left(
l\right)  \ \cdots\ x_{n}\!\left(  l\right)  ,\\
y^{n}\!\left(  m\right)   &  \equiv y_{1}\!\left(  m\right)  \ y_{2}\!\left(
m\right)  \ \cdots\ y_{n}\!\left(  m\right)  .
\end{align*}

\textbf{Transmission}. They both input each letter of their codewords to a
single use of the channel in~(\ref{eq:ccqq-int-channel}), leading to an
$n$-fold tensor product state of the following form at the output:%
\[
\rho_{x^{n}\left(  l\right)  ,y^{n}\left(  m\right)  }^{B_{1}^{n}B_{2}^{n}%
}\equiv\rho_{x_{1}\left(  l\right)  ,y_{1}\left(  m\right)  }^{B_{1,1}B_{2,1}%
}\otimes\rho_{x_{2}\left(  l\right)  ,y_{2}\left(  m\right)  }^{B_{1,2}%
B_{2,2}}\otimes\cdots\otimes\rho_{x_{n}\left(  l\right)  ,y_{n}\left(
m\right)  }^{B_{1,n}B_{2,n}}.
\]
Receiver~1 has access to systems $B_{1,i}$ for all $i \in\{1, \ldots, n\}$,
and Receiver~2 has access to systems $B_{2,i}$.

\textbf{Decoding}. Receiver~1 performs a measurement on his systems in order
to determine the message of Sender~1, and Receiver~2 similarly
performs a measurement to obtain Sender~2's message. More specifically,
Receiver~1 performs a positive operator-valued measure (POVM) $\left\{
\Lambda_{l}\right\}  _{l\in\left\{  1,\ldots,|\mathcal{L}|\right\}  }$ where
$\Lambda_{l}$ is a positive operator for all $l$ and $\sum_{l}\Lambda_{l}=I$,
and Receiver~2 performs a POVM $\left\{  \Gamma_{m}\right\}  _{m\in\left\{
1,\ldots,|\mathcal{M}|\right\}  }$ with similar conditions holding for the
operators in this set. Figure~\ref{fig:info-task}\ depicts all of these steps.

\begin{figure}[ptb]
\begin{center}
\includegraphics[
natheight=1.932900in,
natwidth=4.427000in,
width=3.4705in
]{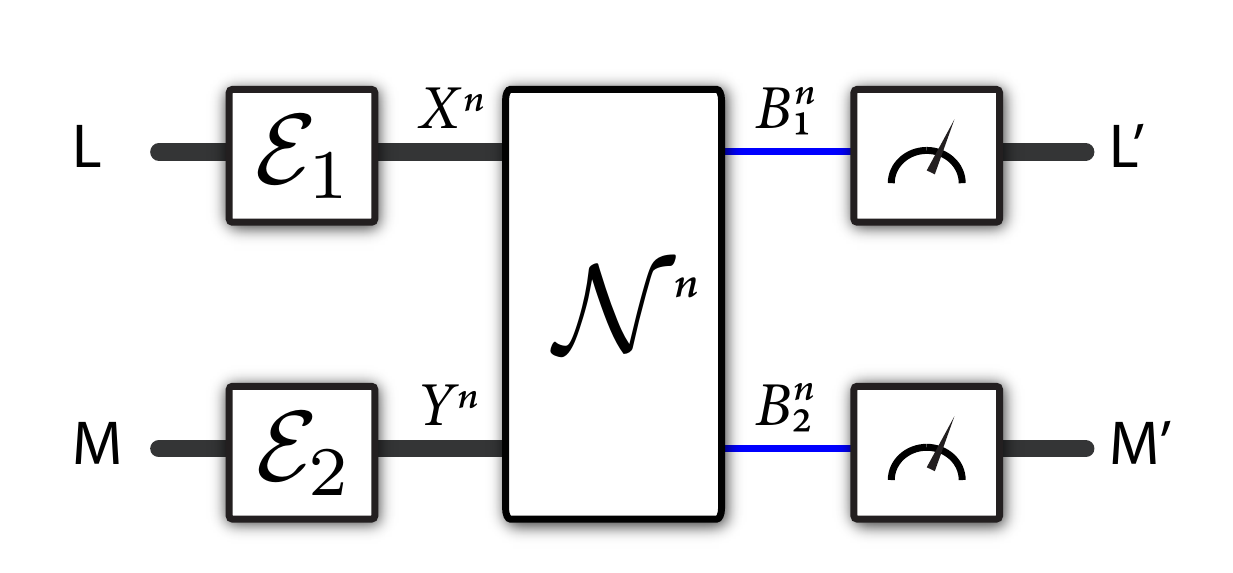}
\end{center}
\caption{The information processing task for the quantum interference channel.
Let $\mathcal{N}$ represent the quantum interference channel with two
classical inputs $X$ and $Y$ and two quantum outputs $B_{1}$ and $B_{2}$.
Sender~1 selects a message $l$ to transmit (modeled by a random variable $L$),
and Sender~2 selects a message $m$ to transmit (modeled by $M$). Each sender
encodes their message as a codeword and transmits the codeword over many
independent uses of a quantum interference channel. The receivers each receive
the quantum outputs of the channel and perform a measurement to determine the
message that their partner sender transmitted.}%
\label{fig:info-task}%
\end{figure}

The probability of the receivers correctly decoding a particular message pair
$\left(  l,m\right)  $ is as follows:
\begin{multline*}
\Pr\left\{  L^{\prime}=l,\ M^{\prime}=m\ |\ L=l,\ M=m\right\} = \\
\text{Tr}\left\{  \left(  \Lambda_{l}\otimes\Gamma_{m}\right)  \rho
_{x^{n}\left(  l\right)  ,y^{n}\left(  m\right)  }^{B_{1}^{n}B_{2}^{n}%
}\right\}  ,
\end{multline*}
and so the probability of incorrectly decoding that message pair is%
\begin{align*}
p_{e}\left(  l,m\right)   &  \equiv\Pr\left\{  (L^{\prime},M^{\prime}%
)\neq(l,m)\ |\ L=l,\ M=m\right\}  \\
&  =\text{Tr}\left\{  \left(  I-\Lambda_{l}\otimes\Gamma_{m}\right)
\rho_{x^{n}\left(  l\right)  ,y^{n}\left(  m\right)  }^{B_{1}^{n}B_{2}^{n}%
}\right\}  ,
\end{align*}
where $L$ and $M$ indicate random variables corresponding to the senders'
choice of messages and the primed random variables correspond to the classical
outputs of the receivers' measurements. The quantum interference channel code
is $\epsilon$-good if the average probability of error $\overline{p}_{e}$ is
bounded from above by $\epsilon$:
\begin{align*}
\overline{p}_{e} &  \equiv\frac{1}{|\mathcal{L}||\mathcal{M}|}\sum_{l,m}%
p_{e}\left(  l,m\right)  \\
&  =\frac{1}{|\mathcal{L}||\mathcal{M}|}\sum_{l,m}\text{Tr}\left\{  \left(
I-\Lambda_{l}\otimes\Gamma_{m}\right)  \rho_{x^{n}\left(  l\right)
,y^{n}\left(  m\right)  }^{B_{1}^{n}B_{2}^{n}}\right\}  \leq\epsilon.
\end{align*}

A rate pair $\left(  R_{1},R_{2}\right)  $ is \textit{achievable} if there
exists an $\left(  n,R_{1}-\delta,R_{2}-\delta,\epsilon\right)  $ quantum
interference channel code for all $\delta,\epsilon>0$ and sufficiently large
$n$. The \textit{capacity region} of the quantum interference channel is the
closure of the set of all achievable rates.

\section{Classical Communication over the Quantum Multiple Access Channel}

\label{sec:MAC}There is a strong connection between the multiple access
channel and the interference channel. In fact, inner bounds for
the capacity of an interference channel can be obtained
by requiring the two receivers to decode both messages.
Such a strategy naturally defines two multiple access channels
that share the same senders~\cite{HK81,el2010lecture}.\footnote{The setting in which
both receivers decode both messages of the two senders is the same as the setting for the
compound multiple access channel \cite{A74}.} It is thus 
important to understand two different coding approaches for obtaining the capacity of
the multiple access channel.

\subsection{Successive Decoding}
\label{sec:mac-succ-decoding}

A first approach to achieve the capacity of the multiple access channel is to exploit a successive
decoding strategy~\cite{book1991cover,el2010lecture}, where the receiver first
decodes the message of one sender while treating the other sender's
transmission as noise. The receiver then decodes the message of the other
sender by exploiting the decoded information as side information. This
strategy achieves one \textquotedblleft corner point\textquotedblright\ of the
capacity region, and a symmetric strategy, where the receiver decodes in the
opposite order, achieves the other corner point. They can achieve any rate
pair between these two corner points with a time-sharing strategy, in which
they exploit successive decoding in one order for a fraction of the channel
uses and they exploit successive decoding in the opposite order for the remaining
fraction of the channel uses. They can achieve the other boundary points and the interior of
the capacity region by resource wasting.

Winter exploited this approach for the quantum multiple access
channel~\cite{winter2001capacity}, essentially by using a random coding
argument and by showing that a measurement to determine the first sender's
message causes a negligible disturbance of the channel output state. Hsieh
\textit{et al}.~followed up on this result by showing how to perform
entanglement-assisted classical communication over a quantum multiple access
channel~\cite{itit2008hsieh}.

\begin{theorem}
[Successive Decoding \cite{winter2001capacity}]\label{thm:successive-decoder}%
Let $x,y\rightarrow\rho_{x,y}$ be a ccq channel from two senders to a single
receiver. Let $p_{X}\left(  x\right)  $ and $p_{Y}\left(  y\right)  $ be
respective input distributions that each sender uses to create random
codebooks of the form $\left\{  X^{n}\left(  l\right)  \right\}  _{l\in\left[
1,\ldots,L\right]  }$ and $\left\{  Y^{n}\left(  m\right)  \right\}
_{m\in\left[  1,\ldots,M\right]  }$. Suppose that the rates $R_{1}=
\frac{1}{n}\log_{2}\left(  L\right)  +\delta $ and $R_{2}=\frac{1}{n}\log
_{2}\left(  M\right)  +\delta $ (where $\delta>0$) satisfy%
\begin{align*}
R_{1}  &  \leq I\left(  X;B\right)  _{\rho},\\
R_{2}  &  \leq I\left(  Y;B|X\right)  _{\rho},
\end{align*}
where the Holevo information quantities are with respect to a
classical-quantum state of the form%
\begin{equation}
\rho^{XYB}\equiv\sum_{x,y}p_{X}\left(  x\right)  p_{Y}\left(  y\right)
\left\vert x\right\rangle \left\langle x\right\vert ^{X}\otimes\left\vert
y\right\rangle \left\langle y\right\vert ^{Y}\otimes\rho_{x,y}^{B}.
\label{eq:state-for-entropies}%
\end{equation}
Then there exist two POVMs $\left\{  \Lambda_{l}\right\}  $ and $\left\{
\Gamma_{m}^{(l)}\right\}  $ acting in successive order such that the expectation of
the average probability of correct detection is arbitrarily close to one:%
\begin{multline*}
\mathbb{E}\left\{  \frac{1}{LM}\sum_{l,m}\text{Tr}\left\{
\sqrt{\Gamma_{m}^{(l)}}\sqrt{\Lambda_{l}}\rho_{X^{n}\left(  l\right)
Y^{n}\left(  m\right)  }\sqrt{\Lambda_{l}}\sqrt{\Gamma_{m}^{(l)}}\right\}
\right\}  \\
\geq1-\epsilon,
\end{multline*}
where the expectation is with respect to $X^{n}$ and $Y^{n}$.

\end{theorem}

\subsection{Quantum Simultaneous Decoding}

Another approach to achieve the capacity of the multiple access channel is for the receiver to 
use a simultaneous decoder (sometimes referred to as a jointly typical
decoder in the IID\ setting), which decodes the messages of all senders
at the same time rather than in succession \cite{book1991cover,el2010lecture}.
On the one hand, simultaneous decoding is more complex than successive decoding
because it considers all tuples
of messages, but on the other hand, it is more powerful than a successive decoding
strategy because it can decode at any rates provided that the rates are in the capacity region
(also there is no need for time sharing).

With such a strategy and for two senders, there are four different types of errors that can
occur---one of these we can bound with a standard typicality argument and the
other three correspond to the bounds on the capacity region of the channel.
This strategy is our approach below, and we can prove that a quantum simultaneous decoder
exists for multiple access channels with two classical inputs and one quantum output.
Though, for a three-sender quantum multiple access channel,
we are only able to prove that a quantum simultaneous decoder
exists in the special case where the averaged output states commute. Thus,
we leave the general case stated as a conjecture.

\subsubsection{Two-Sender Quantum Simultaneous Decoding}

This section contains the proof of the two-sender quantum simultaneous decoder.
We should mention that Sen
arrived at this result with a different technique \cite{S11a}.

\begin{theorem}[Two-Sender Quantum Simultaneous Decoding]\label{thm:two-sender-QSD}
Let $x,y\rightarrow\rho_{x,y}$ be a ccq channel from two senders to a single
receiver. Let $p_{X}\left(  x\right)  $ and $p_{Y}\left(  y\right)  $ be
respective input distributions that each sender uses to create random
codebooks of the form $\left\{  X^{n}\left(  l\right)  \right\}  _{l\in\left[
1,\ldots,L\right]  }$ and $\left\{  Y^{n}\left(  m\right)  \right\}
_{m\in\left[  1,\ldots,M\right]  }$. Suppose that the rates $R_{1}=
\frac{1}{n}\log_{2}\left(  L\right)  +\delta $ and $R_{2}=\frac{1}{n}\log
_{2}\left(  M\right)  +\delta $ (where $\delta>0$) satisfy the following inequalities:
\begin{align}
R_{1}  &  \leq I\left(  X;B|Y\right)  ,\\
R_{2}  &  \leq I\left(  Y;B|X\right)  _{\rho},\\
R_{1}+R_{2}  &  \leq I\left(  XY;B\right)  _{\rho},
\end{align}
where the entropies are with respect to a state of the form in (\ref{eq:state-for-entropies}).
Then there exists a simultaneous decoding POVM $\left\{  \Lambda_{l,m}\right\}  $
such that the expectation of the average probability of error is bounded from above
by $\epsilon$ for all $\epsilon>0$ and sufficiently large $n$.
\end{theorem}

\begin{proof}
Suppose that the channel is a ccq channel of the form $x,y\rightarrow
\rho_{x,y}$ and that the two senders have independent distributions
$p_{X}\left(  x\right)  $ and $p_{Y}\left(  y\right)  $. These distributions
induce the following averaged output states:%
\begin{align}
\rho_{x}  &  \equiv\sum_{y}p_{Y}\!\left(  y\right)  \rho_{x,y}%
,\label{eq:rho_x}\\
\rho_{y}  &  \equiv\sum_{x}p_{X}\!\left(  x\right)  \rho_{x,y}%
,\label{eq:rho_y}\\
\rho &  \equiv\sum_{x,y}p_{X}\!\left(  x\right)  p_{Y}\!\left(  y\right)
\rho_{x,y}. \label{eq:rho}%
\end{align}

\textbf{Codeword Selection}. Senders~1 and 2 choose codewords $\left\{
X^{n}\left(  l\right)  \right\}  _{l\in\left\{  1,\ldots,L\right\}  }$ and
$\left\{  Y^{n}\left(  m\right)  \right\}  _{m\in\left\{  1,\ldots,M\right\}
}$ independently and randomly according to the product distributions $p_{X^{n}}\left(
x^{n}\right)  $ and $p_{Y^{n}}\left(  y^{n}\right)  $.

\textbf{POVM\ Construction}. Let $\Pi_{\rho,\delta}^{n}$ be the typical
projector for the tensor power state $\rho^{\otimes n}$ defined
by~(\ref{eq:rho}). Let $\Pi_{\rho_{y^{n}},\delta}^{n}$ be the conditionally
typical projector for the tensor product state $\rho_{y^{n}}$ defined by
(\ref{eq:rho_y}) for $n$ uses of the channel. Let $\Pi_{\rho_{x^{n}},\delta
}^{n}$ be the conditionally typical projector for the tensor product state
$\rho_{x^{n}}$ defined by (\ref{eq:rho_x}) for $n$ uses of the channel. Let
$\Pi_{\rho_{x^{n},y^{n}},\delta}^{n}$ be the conditionally typical projector
for the tensor product state $\rho_{x^{n},y^{n}}$ defined as the output of the
$n$\ channels when codewords $x^{n}$ and $y^{n}$ are input. (We are using the
\textquotedblleft weak\textquotedblright\ definitions of these projectors as
defined in Appendix~\ref{sec:typ-review}.) In what follows, we make the following
abbreviations:%
\begin{align*}
\Pi &  \equiv\Pi_{\rho,\delta}^{n},\\
\Pi_{y^{n}}  &  \equiv\Pi_{\rho_{y^{n}},\delta}^{n},\\
\Pi_{x^{n}}  &  \equiv\Pi_{\rho_{x^{n}},\delta}^{n},\\
\Pi_{x^{n},y^{n}}  &  \equiv\Pi_{\rho_{x^{n},y^{n}},\delta}^{n}.
\end{align*}
The detection POVM\ $\left\{  \Lambda_{l,m}\right\}  $ has the following form:%
\begin{align}
\Lambda_{l,m}  &  \equiv\left(  \sum_{l^{\prime},m^{\prime}}\Pi_{l^{\prime
},m^{\prime}}^{\prime}\right)  ^{-\frac{1}{2}}\Pi_{l,m}^{\prime}\left(
\sum_{l^{\prime},m^{\prime}}\Pi_{l^{\prime},m^{\prime}}^{\prime}\right)
^{-\frac{1}{2}},\label{eq:square-root-POVM}\\
\Pi_{l,m}^{\prime}  &  \equiv\Pi\ \Pi_{X^{n}\left(  l\right)  }\ \Pi
_{X^{n}\left(  l\right)  ,Y^{n}\left(  m\right)  }\ \Pi_{X^{n}\left(
l\right)  }\ \Pi.\nonumber
\end{align}
Observe that the operator $\Pi_{l,m}^{\prime}$ is a positive operator and
thus $\left\{  \Lambda_{l,m}\right\}  $ is a valid POVM.

\textbf{Error Analysis}. The average error probability of the code has the
following form:%
\begin{equation}
\overline{p}_{e}\equiv\frac{1}{LM}\sum_{l,m}\text{Tr}\left\{  \left(
I-\Lambda_{l,m}\right)  \rho_{X^{n}\left(  l\right)  ,Y^{n}\left(  m\right)
}\right\}  .\label{eq:avg-error-prob}%
\end{equation}
We instead analyze the expectation of the average error probability, where the
expectation is with respect to the random choice of code:%
\begin{align*}
&  \mathbb{E}_{X^{n},Y^{n}}\left\{  \overline{p}_{e}\right\}  \\
&  \equiv\mathbb{E}_{X^{n},Y^{n}}\left\{  \frac{1}{LM}\sum_{l,m}%
\text{Tr}\left\{  \left(  I-\Lambda_{l,m}\right)  \rho_{X^{n}\left(  l\right)
,Y^{n}\left(  m\right)  }\right\}  \right\}  \\
&  =\frac{1}{LM}\sum_{l,m}\mathbb{E}_{X^{n},Y^{n}}\left\{  \text{Tr}\left\{
\left(  I-\Lambda_{l,m}\right)  \rho_{X^{n}\left(  l\right)  ,Y^{n}\left(
m\right)  }\right\}  \right\}  .
\end{align*}
Due to the symmetry of the code construction (the fact that the expectation
$\mathbb{E}_{X^{n},Y^{n}}\left\{  \text{Tr}\left\{  \left(  I-\Lambda
_{l,m}\right)  \rho_{X^{n}\left(  l\right)  ,Y^{n}\left(  m\right)  }\right\}
\right\}  $ is independent of the particular message pair $\left(  l,m\right)
$), it suffices to analyze the expectation of the average error probability
for the first message pair $\left(  1,1\right)  $:%
\[
\mathbb{E}_{X^{n},Y^{n}}\left\{  \overline{p}_{e}\right\}  =\mathbb{E}%
_{X^{n},Y^{n}}\left\{  \text{Tr}\left\{  \left(  I-\Lambda_{1,1}\right)
\rho_{X^{n}\left(  1\right)  ,Y^{n}\left(  1\right)  }\right\}  \right\}  .
\]

We now begin our error analysis. In what follows, we abbreviate $X^{n}$ as $X$
and $Y^{n}$ as $Y$ in order to save space. We first bound the above error
probability as%
\begin{align}
&  \mathbb{E}_{XY}\left\{  \overline{p}_{e}\right\}  \nonumber\\
&  \leq\mathbb{E}_{XY}\left\{  \text{Tr}\left\{  \left(  I-\Lambda
_{1,1}\right)  \ \Pi_{Y\left(  1\right)  }\ \rho_{X\left(  1\right)  ,Y\left(
1\right)  }\ \Pi_{Y\left(  1\right)  }\right\}  \right\}  \nonumber \\
&  \ \ +\mathbb{E}_{XY}\left\{  \left\Vert \Pi_{Y\left(  1\right)  }%
\ \rho_{X\left(  1\right)  ,Y\left(  1\right)  }\ \Pi_{Y\left(  1\right)
}-\rho_{X\left(  1\right)  ,Y\left(  1\right)  }\right\Vert _{1}\right\}  \\
&  \leq\mathbb{E}_{XY}\left\{  \text{Tr}\left\{  \left(  I-\Lambda
_{1,1}\right)  \ \Pi_{Y\left(  1\right)  }\ \rho_{X\left(  1\right)  ,Y\left(
1\right)  }\ \Pi_{Y\left(  1\right)  }\right\}  \right\}  \nonumber \\
& \,\,\, +2\sqrt{\epsilon
},\label{eq:first-bound}%
\end{align}
where the first inequality follows from the inequality%
\begin{equation}
\text{Tr}\left\{  \Lambda\rho\right\}  \leq\text{Tr}\left\{  \Lambda
\sigma\right\}  +\left\Vert \rho-\sigma\right\Vert _{1}%
,\label{eq:trace-inequality}%
\end{equation}
which holds for all $\rho$, $\sigma$, and $\Lambda$ such that $0\leq
\rho,\sigma,\Lambda\leq I$. The second inequality follows from the properties
of weak conditionally typical subspaces and the Gentle Operator Lemma for
ensembles, by taking $n$ to be sufficiently large (a discussion of these
properties is in Appendix~\ref{sec:typ-review}).
The idea behind this first bound on the error
probability is that we require the projector $\Pi_{Y\left(  1\right)}$ in order to
remove some of large eigenvalues of an averaged version of $\rho_{X\left(  1\right)  ,Y\left(
1\right)}$, and this point in the proof seems to be the most opportune time to insert it.

The Hayashi-Nagaoka operator inequality applies to a positive operator $T$ and
an operator $S$ where $0\leq S\leq I$ \cite{hayashi2003general,H06book}:%
\[
I-\left(  S+T\right)  ^{-\frac{1}{2}}S\left(  S+T\right)  ^{-\frac{1}{2}}%
\leq2\left(  I-S\right)  +4T.
\]
Choosing%
\begin{align*}
S &  =\Pi_{1,1}^{\prime},\\
T &  =\sum_{\left(  l,m\right)  \neq\left(  1,1\right)  }\Pi_{l,m}^{\prime},
\end{align*}
we can apply the above operator inequality to bound the first term in (\ref{eq:first-bound}) as%
\begin{multline}
\mathbb{E}_{XY}\left\{  \text{Tr}\left\{  \left(  I-\Lambda_{1,1}\right)
\ \Pi_{Y\left(  1\right)  }\ \rho_{X\left(  1\right)  ,Y\left(  1\right)
}\ \Pi_{Y\left(  1\right)  }\right\}  \right\}  \label{eq:error-bound-HN}\\
\leq2\ \mathbb{E}_{XY}\left\{  \text{Tr}\left\{  \left(  I-\Pi_{1,1}^{\prime
}\right)  \ \Pi_{Y\left(  1\right)  }\ \rho_{X\left(  1\right)  ,Y\left(
1\right)  }\ \Pi_{Y\left(  1\right)  }\right\}  \right\}  \\
\qquad \ \ +4\!\!\!\!\!\!\sum_{\left(  l,m\right)  \neq\left(  1,1\right)  }\mathbb{E}_{XY}\left\{
\text{Tr}\left\{  \Pi_{l,m}^{\prime}\ \Pi_{Y\left(  1\right)  }\ \rho
_{X\left(  1\right)  ,Y\left(  1\right)  }\ \Pi_{Y\left(  1\right)  }\right\}
\right\}  .
\end{multline}
We first consider bounding the term in the second line above. Consider that%
\begin{align}
&  \mathbb{E}_{XY}\left\{  \text{Tr}\left\{  \Pi_{1,1}^{\prime}\ \Pi_{Y\left(
1\right)  }\ \rho_{X\left(  1\right)  ,Y\left(  1\right)  }\ \Pi_{Y\left(
1\right)  }\right\}  \right\}  \nonumber\\
&  =\mathbb{E}_{XY}\{\text{Tr}\{\Pi\ \Pi_{X\left(  1\right)  }\ \Pi_{X\left(
1\right)  ,Y\left(  1\right)  }\ \Pi_{X\left(  1\right)  }\ \Pi\nonumber\\
&  \ \ \ \ \ \ \ \ \ \ \ \ \ \ \ \ \ \ \ \ \ \ \ \Pi_{Y\left(  1\right)
}\ \rho_{X\left(  1\right)  ,Y\left(  1\right)  }\ \Pi_{Y\left(  1\right)
}\}\}\nonumber\\
&  \geq\mathbb{E}_{XY}\left\{  \text{Tr}\left\{  \Pi_{X\left(  1\right)
,Y\left(  1\right)  }\ \rho_{X\left(  1\right)  ,Y\left(  1\right)  }\right\}
\right\}  \nonumber\\
&  \ \ \ -\mathbb{E}_{XY}\left\{  \left\Vert \Pi\ \rho_{X\left(  1\right)
,Y\left(  1\right)  }\ \Pi-\rho_{X\left(  1\right)  ,Y\left(  1\right)
}\right\Vert _{1}\right\}  \nonumber\\
&  \ \ \ \ -\mathbb{E}_{XY}\left\{  \left\Vert \Pi_{Y\left(  1\right)  }%
\ \rho_{X\left(  1\right)  ,Y\left(  1\right)  }\ \Pi_{Y\left(  1\right)
}-\rho_{X\left(  1\right)  ,Y\left(  1\right)  }\right\Vert _{1}\right\}
\nonumber\\
&  \ \ \ \  \  -\mathbb{E}_{XY}\left\{  \left\Vert \Pi_{X\left(  1\right)  }%
\ \rho_{X\left(  1\right)  ,Y\left(  1\right)  }\ \Pi_{X\left(  1\right)
}-\rho_{X\left(  1\right)  ,Y\left(  1\right)  }\right\Vert _{1}\right\}
\nonumber\\
&  \geq1-\epsilon-6\sqrt{\epsilon}.\label{eq:first-error-chain}%
\end{align}
The above inequalities follow by employing the Gentle Operator Lemma for
ensembles, (\ref{eq:trace-inequality}), and the below inequalities that follow
from the discussion in Appendix~\ref{sec:typ-review}:%
\begin{align}
\mathbb{E}_{XY}\left\{  \text{Tr}\{\Pi_{X\left(  1\right)  }\ \rho_{X\left(
1\right)  ,Y\left(  1\right)  }\}\right\}   &  \geq1-\epsilon
,\label{eq:typ-cond-typ-1}\\
\mathbb{E}_{XY}\left\{  \text{Tr}\{\Pi_{Y\left(  1\right)  }\ \rho_{X\left(
1\right)  ,Y\left(  1\right)  }\}\right\}   &  \geq1-\epsilon
,\label{eq:typ-cond-typ-2}\\
\mathbb{E}_{XY}\left\{  \text{Tr}\{\Pi\ \rho_{X\left(  1\right)  ,Y\left(
1\right)  }\}\right\}   &  \geq1-\epsilon.\label{eq:typ-cond-typ-3}\\
\mathbb{E}_{XY}\left\{  \text{Tr}\{\Pi_{X\left(  1\right)  ,Y\left(  1\right)
}\ \rho_{X\left(  1\right)  ,Y\left(  1\right)  }\}\right\}   &
\geq1-\epsilon.\label{eq:typ-cond-typ-4}%
\end{align}
This bound then implies that%
\begin{equation}
\mathbb{E}_{XY}\left\{  \text{Tr}\left\{  \left(  I-\Pi_{1,1}^{\prime}\right)
\ \Pi_{Y\left(  1\right)  }\ \rho_{X\left(  1\right)  ,Y\left(  1\right)
}\ \Pi_{Y\left(  1\right)  }\right\}  \right\}  \leq\epsilon+6\sqrt{\epsilon
}.\label{eq:first-type-bound}%
\end{equation}
The bound in (\ref{eq:error-bound-HN}) reduces to the following one after
applying (\ref{eq:first-type-bound}):%
\begin{multline*}
\mathbb{E}_{XY}\left\{\overline{p}_{e}\right\}\leq2\left(  \epsilon+6\sqrt{\epsilon}\right)  \\
+4\!\!\!\!\!\!\sum_{\left(  l,m\right)  \neq\left(  1,1\right)  }\!\!\!\mathbb{E}_{XY}\left\{
\text{Tr}\left\{  \Pi_{l,m}^{\prime}\ \Pi_{Y\left(  1\right)  }\ \rho
_{X\left(  1\right)  ,Y\left(  1\right)  }\ \Pi_{Y\left(  1\right)  }\right\}
\right\}  .
\end{multline*}
We can expand the doubly-indexed sum in the above expression as%
\begin{multline}
\sum_{\left(  l,m\right)  \neq\left(  1,1\right)  }\mathbb{E}_{XY}\left\{
\text{Tr}\left\{  \Pi_{l,m}^{\prime}\ \Pi_{Y\left(  1\right)  }\ \rho
_{X\left(  1\right)  ,Y\left(  1\right)  }\ \Pi_{Y\left(  1\right)  }\right\}
\right\}  =\\
\sum_{l\neq1}\mathbb{E}_{XY}\left\{  \text{Tr}\left\{  \Pi_{l,1}^{\prime}%
\ \Pi_{Y\left(  1\right)  }\ \rho_{X\left(  1\right)  ,Y\left(  1\right)
}\ \Pi_{Y\left(  1\right)  }\right\}  \right\}  \\
\ \ \ +\sum_{m\neq1}\mathbb{E}_{XY}\left\{  \text{Tr}\left\{  \Pi_{1,m}^{\prime
}\ \Pi_{Y\left(  1\right)  }\ \rho_{X\left(  1\right)  ,Y\left(  1\right)
}\ \Pi_{Y\left(  1\right)  }\right\}  \right\}  \\
\ \ \ \ +\!\!\!\!\!\sum_{l\neq1,\ m\neq1}\!\!\!\!\!\mathbb{E}_{XY}\left\{  \text{Tr}\left\{  \Pi
_{l,m}^{\prime}\ \Pi_{Y\left(  1\right)  }\ \rho_{X\left(  1\right)  ,Y\left(
1\right)  }\ \Pi_{Y\left(  1\right)  }\right\}  \right\}
.\label{eq:three-terms}%
\end{multline}
We begin by bounding the term in the second line above. Consider the following
chain of inequalities:%
\begin{align}
&  \sum_{l\neq1}\mathbb{E}_{XY}\left\{  \text{Tr}\left\{  \Pi_{l,1}^{\prime
}\Pi_{Y\left(  1\right)  }\ \rho_{X\left(  1\right)  ,Y\left(  1\right)
}\ \Pi_{Y\left(  1\right)  }\right\}  \right\}  \nonumber\\
&  =\sum_{l\neq1}\mathbb{E}_{Y}\left\{  \text{Tr}\left\{  \mathbb{E}%
_{X}\left\{  \Pi_{l,1}^{\prime}\right\}  \ \Pi_{Y\left(  1\right)  }%
\mathbb{E}_{X}\left\{  \rho_{X\left(  1\right)  ,Y\left(  1\right)  }\right\}
\Pi_{Y\left(  1\right)  }\right\}  \right\}  \nonumber\\
&  =\sum_{l\neq1}\mathbb{E}_{Y}\left\{  \text{Tr}\left\{  \mathbb{E}%
_{X}\left\{  \Pi_{l,1}^{\prime}\right\}  \ \Pi_{Y\left(  1\right)  }%
\ \rho_{Y\left(  1\right)  }\ \Pi_{Y\left(  1\right)  }\right\}  \right\}
\nonumber\\
&  \leq2^{-n\left[  H\left(  B|Y\right)  -\delta\right]  }\sum_{l\neq
1}\mathbb{E}_{Y}\left\{  \text{Tr}\left\{  \mathbb{E}_{X}\left\{  \Pi
_{l,1}^{\prime}\right\}  \ \Pi_{Y\left(  1\right)  }\right\}  \right\}
\nonumber\\
&  =2^{-n\left[  H\left(  B|Y\right)  -\delta\right]  }\sum_{l\neq1}%
\mathbb{E}_{XY}\left\{  \text{Tr}\left\{  \Pi_{l,1}^{\prime}\ \Pi_{Y\left(
1\right)  }\right\}  \right\} \label{eq:1st-error-term}
\end{align}
The first equality follows because $X\left(  l\right)  $ and $X\left(
1\right)  $ are independent---the senders choose the code randomly in such a
way that this is true. The second equality follows because $\mathbb{E}%
_{X}\left\{  \rho_{X\left(  1\right)  ,Y\left(  1\right)  }\right\}
=\rho_{Y\left(  1\right)  }$. The first inequality follows by applying the
following operator inequality for weak conditionally typical subspaces:%
\[
\Pi_{y^{n}}\ \rho_{y^{n}}\ \Pi_{y^{n}}\leq2^{-n\left[  H\left(  B|Y\right)
-\delta\right]  }\ \Pi_{y^{n}}.
\]
The last equality is from factoring out the expectation. We now focus on the
expression inside the expectation:%
\begin{align*}
& \text{Tr}\left\{  \Pi_{l,1}^{\prime}\ \Pi_{Y\left(  1\right)  }\right\}  \\
& \quad =\text{Tr}\left\{  \Pi\ \Pi_{X\left(  l\right)  }\ \Pi_{X\left(  l\right)
,Y\left(  1\right)  }\ \Pi_{X\left(  l\right)  }\ \Pi\ \Pi_{Y\left(  1\right)
}\right\}  \\
&\quad  =\text{Tr}\left\{  \Pi_{X\left(  l\right)  ,Y\left(  1\right)  }%
\ \Pi_{X\left(  l\right)  }\ \Pi\ \Pi_{Y\left(  1\right)  }\ \Pi
\ \Pi_{X\left(  l\right)  }\right\}  \\
&\quad  \leq\text{Tr}\left\{  \Pi_{X\left(  l\right)  ,Y\left(  1\right)  }\right\}
\\
& \quad \leq2^{n\left[  H\left(  B|XY\right)  +\delta\right]  }.
\end{align*}
The first equality is from substitution. The second equality is from cyclicity
of trace. The first inequality is from%
\[
\Pi_{x^{n}}\ \Pi\ \Pi_{y^{n}}\ \Pi\ \Pi_{x^{n}}\leq\Pi_{x^{n}}\ \Pi
\ \Pi_{x^{n}}\leq\Pi_{x^{n}}\leq I.
\]
The final inequality follows from the bound on the rank of the weak conditionally
typical projector (see Appendix~\ref{sec:typ-review}).

Substituting back into (\ref{eq:1st-error-term}), we have%
\begin{align*}
& \sum_{l\neq1}\mathbb{E}_{XY}\left\{  \text{Tr}\left\{  \Pi_{l,1}^{\prime}%
\Pi_{Y\left(  1\right)  }\ \rho_{X\left(  1\right)  ,Y\left(  1\right)  }%
\ \Pi_{Y\left(  1\right)  }\right\}  \right\}  \\
&\qquad \leq2^{-n\left[  H\left(  B|Y\right)  -\delta\right]  }\sum_{l\neq
1}2^{n\left[  H\left(  B|XY\right)  +\delta\right]  }\\
&\qquad \leq2^{-n\left[  H\left(  B|Y\right)  -\delta\right]  }\ 2^{n\left[
H\left(  B|XY\right)  +\delta\right]  }\ L\\
& \qquad=2^{-n\left[  I\left(  X;B|Y\right)  -2\delta\right]  }\ L.
\end{align*}

We employ a different argument to bound the term in the third line of
(\ref{eq:three-terms}). Consider the following chain of inequalities:%
\begin{align}
& \sum_{m\neq1}\mathbb{E}_{XY}\left\{  \text{Tr}\left\{  \Pi_{1,m}^{\prime}%
\ \Pi_{Y\left(  1\right)  }\ \rho_{X\left(  1\right)  ,Y\left(  1\right)
}\ \Pi_{Y\left(  1\right)  }\right\}  \right\}  \nonumber \\
& \quad=\sum_{m\neq1}\mathbb{E}_{X}\{  \text{Tr}\{  \mathbb{E}_{Y}\{
\Pi_{1,m}^{\prime}\}  \mathbb{E}_{Y}\{  \Pi_{Y\left(  1\right)
}\ \rho_{X\left(  1\right)  ,Y\left(  1\right)  }\ \Pi_{Y\left(  1\right)
}\}  \}  \}  \label{eq:second-error}%
\end{align}
This equality follows from the fact that $Y\left(  m\right)  $ and $Y\left(
1\right)  $ are independent. We now focus on bounding the operator $\mathbb{E}%
_{Y}\left\{  \Pi_{1,m}^{\prime}\right\}  $ inside the trace:%
\begin{align}
 \mathbb{E}_{Y}&\!\left\{  \Pi_{1,m}^{\prime}\right\}  \nonumber\\
& =\mathbb{E}_{Y}\left\{  \Pi\ \Pi_{X\left(  1\right)  }\ \Pi_{X\left(
1\right)  ,Y\left(  m\right)  }\ \Pi_{X\left(  1\right)  }\ \Pi\right\}
\nonumber\\
& \leq2^{n\left[  H\left(  B|XY\right)  +\delta\right]  }\ \mathbb{E}%
_{Y}\left\{  \Pi\ \Pi_{X\left(  1\right)  }\ \rho_{X\left(  1\right)
,Y\left(  m\right)  }\ \Pi_{X\left(  1\right)  }\ \Pi\right\}  \nonumber\\
& =2^{n\left[  H\left(  B|XY\right)  +\delta\right]  }\ \Pi\ \Pi_{X\left(
1\right)  }\ \mathbb{E}_{Y}\left\{  \rho_{X\left(  1\right)  ,Y\left(
m\right)  }\right\}  \ \Pi_{X\left(  1\right)  }\ \Pi\nonumber\\
& =2^{n\left[  H\left(  B|XY\right)  +\delta\right]  }\ \Pi\ \Pi_{X\left(
1\right)  }\ \rho_{X\left(  1\right)  }\ \Pi_{X\left(  1\right)  }%
\ \Pi\nonumber\\
& \leq2^{n\left[  H\left(  B|XY\right)  +\delta\right]  }\ 2^{-n\left[
H\left(  B|X\right)  -\delta\right]  }\ \Pi\ \Pi_{X\left(  1\right)  }%
\ \Pi\nonumber\\
& =2^{-n\left[  I\left(  Y;B|X\right)  -2\delta\right]  }\ \Pi\ \Pi_{X\left(
1\right)  }\ \Pi \nonumber \\
& =2^{-n\left[  I\left(  Y;B|X\right)  -2\delta\right]  }\ I \label{eq:second-type-error}%
\end{align}
The first equality follows by substitution. The first inequality follows from
the following operator inequality:%
\begin{align*}
\Pi_{x^{n},y^{n}}  & \leq2^{n\left[  H\left(  B|XY\right)  +\delta\right]
}\ \Pi_{x^{n},y^{n}}\ \rho_{x^{n},y^{n}}\ \Pi_{x^{n},y^{n}}\\
& = 2^{n\left[  H\left(  B|XY\right)  +\delta\right]
}\ \Pi_{x^{n},y^{n}}\ \sqrt{\rho_{x^{n},y^{n}}} \sqrt{\rho_{x^{n},y^{n}}} \ \Pi_{x^{n},y^{n}}\\
& = 2^{n\left[  H\left(  B|XY\right)  +\delta\right]
}\ \sqrt{\rho_{x^{n},y^{n}}} \Pi_{x^{n},y^{n}} \sqrt{\rho_{x^{n},y^{n}}} \\
& \leq2^{n\left[  H\left(  B|XY\right)  +\delta\right]  }\ \rho_{x^{n},y^{n}}.
\end{align*}
The second equality follows because $\Pi\ $\ and $\Pi_{X\left(  1\right)  }$
are constants with respect to the expectation over $Y$. The third equality
follows because $\mathbb{E}_{Y}\left\{  \rho_{X\left(  1\right)  ,Y\left(
m\right)  }\right\}  =\rho_{X\left(  1\right)  }$, and the second inequality
follows from the operator inequality%
\[
\Pi_{x^{n}}\ \rho_{x^{n}}\ \Pi_{x^{n}}\leq2^{-n\left[  H\left(  B|X\right)
-\delta\right]  }\Pi_{x^{n}}.
\]
The final inequality follows from
\[
\Pi\ \Pi_{x^n }\ \Pi \leq \Pi \leq I .
\]
Substituting the operator inequality in (\ref{eq:second-type-error}) into
(\ref{eq:second-error}), we have%
\begin{align*}
& \sum_{m\neq1}\mathbb{E}_{XY}\left\{  \text{Tr}\left\{  \Pi_{1,m}^{\prime
}\ \Pi_{Y\left(  1\right)  }\ \rho_{X\left(  1\right)  ,Y\left(  1\right)
}\ \Pi_{Y\left(  1\right)  }\right\}  \right\}  \\
& \ \ 
 \leq2^{-n\left[  I\left(  Y;B|X\right)  -2\delta\right]  }
\!\sum_{m\neq1}\mathbb{E}_{XY}\!\!\left\{  \text{Tr}\left\{    \Pi_{Y\left(  1\right)  }
 \rho_{X\left(  1\right)  ,Y\left(  1\right)  } \Pi_{Y\left(  1\right)
}  \right\}  \right\}  \\
& \ \  \leq2^{-n\left[  I\left(  Y;B|X\right)  -2\delta\right]  }\sum_{m\neq
1}\mathbb{E}_{XY}\left\{  \text{Tr}\left\{  \rho_{X\left(  1\right)  ,Y\left(
1\right)  }\right\}  \right\}  \\
& \ \ 
\leq2^{-n\left[  I\left(  Y;B|X\right)  -2\delta\right]  }\ M
\end{align*}
The second inequality follows because $\Pi_{y^{n}} \leq I$.

Finally, we obtain a bound on the term in the last line of (\ref{eq:three-terms}) with a
slightly different argument:%
\begin{align*}
&  \!\!\!\!\!\sum_{\ \ l\neq1,m\neq1}\mathbb{E}_{XY}\left\{  \text{Tr}\left\{  \Pi
_{l,m}^{\prime}\ \Pi_{Y\left(  1\right)  }\ \rho_{X\left(  1\right)  ,Y\left(
1\right)  }\ \Pi_{Y\left(  1\right)  }\right\}  \right\}  \\
&  =\sum_{\substack{l\neq1,\\m\neq1}}\mathbb{E}_{Y}\left\{  \text{Tr}\left\{
\mathbb{E}_{X}\left\{  \Pi_{l,m}^{\prime}\right\}  \Pi_{Y\left(  1\right)
}\mathbb{E}_{X}\left\{  \rho_{X\left(  1\right)  ,Y\left(  1\right)
}\right\}  \Pi_{Y\left(  1\right)  }\right\}  \right\}  \\
&  =\!\!\!\!\!\sum_{l\neq1,\ m\neq1}\mathbb{E}_{Y}\left\{  \text{Tr}\left\{
\mathbb{E}_{X}\left\{  \Pi_{l,m}^{\prime}\right\}  \ \Pi_{Y\left(  1\right)
}\ \rho_{Y\left(  1\right)  }\ \Pi_{Y\left(  1\right)  }\right\}  \right\}  \\
&  \leq\!\!\!\!\!\sum_{l\neq1,\ m\neq1}\mathbb{E}_{Y}\left\{  \text{Tr}\left\{
\mathbb{E}_{X}\left\{  \Pi_{l,m}^{\prime}\right\}  \ \rho_{Y\left(  1\right)
}\right\}  \right\}  \\
&  =\sum_{\substack{l\neq1,\\m\neq1}}\mathbb{E}_{XY}\left\{  \text{Tr}\left\{
\Pi\ \Pi_{X\left(  l\right)  }\ \Pi_{X\left(  l\right)  ,Y\left(  m\right)
}\ \Pi_{X\left(  l\right)  }\ \Pi\ \rho_{Y\left(  1\right)  }\right\}
\right\}  \\
&  =\sum_{\substack{l\neq1,\\m\neq1}}\mathbb{E}_{X}\{\text{Tr}\{\Pi
\Pi_{X\left(  l\right)  }\mathbb{E}_{Y}\{\Pi_{X\left(  l\right)  Y\left(
m\right)  }\}\Pi_{X\left(  l\right)  }\Pi\mathbb{E}_{Y}\{\rho_{Y\left(
1\right)  }\}\}\}\\
&  =\sum_{\substack{l\neq1,\\m\neq1}}\mathbb{E}_{X}\left\{  \text{Tr}\left\{
\Pi_{X\left(  l\right)  }\ \mathbb{E}_{Y}\left\{  \Pi_{X\left(  l\right)
,Y\left(  m\right)  }\right\}  \ \Pi_{X\left(  l\right)  }\ \Pi\ \rho^{\otimes
n}\ \Pi\right\}  \right\}
\end{align*}
The first equality follows from the independence of $X\left(  l\right)  $ and
$X\left(  1\right)  $. The second equality follows because $\mathbb{E}%
_{X}\left\{  \rho_{X\left(  1\right)  ,Y\left(  1\right)  }\right\}
=\rho_{Y\left(  1\right)  }$. The first inequality follows from the fact that
$\rho_{y^{n}}$ and $\Pi_{y^{n}}$ commute and thus $\Pi_{y^{n}}\ \rho_{y^{n}%
}\ \Pi_{y^{n}}=\sqrt{\rho_{y^{n}}}\ \Pi_{y^{n}}\ \sqrt{\rho_{y^{n}}}\leq
\rho_{y^{n}}$. The third equality follows from factoring out the expectation
and substitution of the definition of $\Pi_{l,m}^{\prime}$. The fourth
equality follows from the independence of $Y\left(  m\right)  $ and $Y\left(
1\right)  $. The last equality follows because $\mathbb{E}_{Y}\left\{
\rho_{Y\left(  1\right)  }\right\}  =\rho^{\otimes n}$ and from cyclicity of
trace. Continuing, we have%
\begin{align}
&  \leq2^{-n\left[  H\left(  B\right)  -\delta\right]  }\times\nonumber\\
&  \ \ \sum_{l\neq1,\ m\neq1}\text{Tr}\left\{  \mathbb{E}_{X}\left\{
\Pi_{X\left(  l\right)  }\ \mathbb{E}_{Y}\left\{  \Pi_{X\left(  l\right)
,Y\left(  m\right)  }\right\}  \ \Pi_{X\left(  l\right)  }\right\}
\ \Pi\right\}  \nonumber
\end{align}
\begin{align}
&  =2^{-n\left[  H\left(  B\right)  -\delta\right]  }\times\nonumber\\
&  \ \sum_{l\neq1,\ m\neq1}\mathbb{E}_{XY}\left\{  \text{Tr}\left\{
\Pi_{X\left(  l\right)  ,Y\left(  m\right)  }\ \Pi_{X\left(  l\right)  }%
\ \Pi\ \Pi_{X\left(  l\right)  }\right\}  \right\}  \nonumber\\
&  \leq2^{-n\left[  H\left(  B\right)  -\delta\right]  }\sum_{l\neq1,\ m\neq
1}\mathbb{E}_{XY}\left\{  \text{Tr}\left\{  \Pi_{X\left(  l\right)  ,Y\left(
m\right)  }\right\}  \right\}  \nonumber\\
&  \leq2^{-n\left[  H\left(  B\right)  -\delta\right]  }\ 2^{n\left[  H\left(
B|XY\right)  +\delta\right]  }\ LM\nonumber\\
&  =2^{-n\left[  I\left(  XY;B\right)  -2\delta\right]  }\ LM.
\end{align}
The first inequality is from the following operator inequality:%
\[
\Pi\ \rho^{\otimes n}\ \Pi\leq2^{-n\left[  H\left(  B\right)  -\delta\right]
}\Pi.
\]
The second equality is from cyclicity of trace and factoring out the
expectations. The second inequality is from the operator inequality%
\[
\Pi_{x^{n}}\ \Pi\ \Pi_{x^{n}}\leq\Pi_{x^{n}}\leq I.
\]
The final inequality is from the bound on the rank of the weak conditionally
typical projector.

Combining everything together, we get the following bound on the expectation
of the average error probability:%
\begin{multline*}
\mathbb{E}_{XY}\left\{  \overline{p}_{e}\right\}  \leq2\left(  \epsilon
+7\sqrt{\epsilon}\right)  +\\
4\ L\ 2^{-n\left[  I\left(  X;B|Y\right)  -2\delta\right]  }%
+4\ M\ 2^{-n\left[  I\left(  Y;B|X\right)  -2\delta\right]  }+\\
4\ LM\ 2^{-n\left[  I\left(  XY;B\right)  -2\delta\right]  }.
\end{multline*}
Thus, we can choose the message sizes to be as follows:%
\begin{align*}
L &  =2^{n\left[  R_{1}-3\delta\right]  },\\
M &  =2^{n\left[  R_{2}-3\delta\right]  },
\end{align*}
so that the expectation of the average error probability vanishes in the
asymptotic limit whenever the rates $R_{1}$ and $R_{2}$ obey the following
inequalities:%
\begin{align*}
R_{1}-\delta &  <I\left(  X;B|Y\right)  ,\\
R_{2}-\delta &  <I\left(  Y;B|X\right)  ,\\
R_{1}+R_{2}-4\delta &  <I\left(  XY;B\right)  .
\end{align*}

\end{proof}

A casual glance at the above proof might lead one to believe it is just a
straightforward extension of the ``usual'' proofs of the HSW theorem
\cite{ieee1998holevo,PhysRevA.56.131,ieee2005dev,itit2008hsieh,W11}, but it
differs from these and extends them non trivially in several regards. First,
we choose the square-root POVM in (\ref{eq:square-root-POVM}) in a particular
way---specifically, the layering of projectors is such that the projector of
size $\approx2^{n H(B|XY)}$ is surrounded by the projector of size
$\approx2^{H(B|X)}$, which itself is surrounded by the projector of size
$\approx2^{nH(B)}$. If one were to place the projector of size $\approx
2^{nH(B|Y)}$ somewhere in the square-root POVM, this leads to difficulties
with non-commutative projectors (discussed in earlier versions of this paper
on the arXiv). So, our second observation is to instead ``smooth'' the state
by the projector of size $\approx2^{nH(B|Y)}$ before applying the
Hayashi-Nagaoka operator inequality. The above combination seems to be just
the right trick for applying independence of the codewords after invoking the
Hayashi-Nagaoka operator inequality. The final way in which our proof differs
from earlier ones is that we analyze each of the four errors in a different
way (these four types of errors occur after the application of the
Hayashi-Nagaoka operator inequality). This asymmetry does not occur in the
error analysis of the classical multiple access channel (see page 4-15 of
Ref.~\cite{el2010lecture}), but for the moment, it seems to be necessary in
the quantum case due to the general non-commutativity of typical projectors.
Many of these observations are present in Sen's proof of the above
theorem \cite{S11a}, but his proof introduces several new techniques
(interestingly, he does not exploit the familiar square-root POVM or the
Hayashi-Nagaoka operator inequality).

We obtain the following simple corollary of Theorem~\ref{thm:two-sender-QSD}
by a technique called ``coded time-sharing'' \cite{HK81,el2010lecture}. The
main idea is to pick a sequence $q^{n}$ according to a product distribution
$p_{Q^{n}}(q^{n})$ and then pick the codeword sequences $x^{n}$ and $y^{n}$
according to $p_{X^{n}|Q^{n}}(x^{n}|q^{n})$ and $p_{Y^{n}|Q^{n}}(y^{n}|q^{n}%
)$, respectively (so that $x^{n}$ and $y^{n}$ are conditionally independent
when given $q^{n}$). In the proof, all typical projectors are conditional on
$q^{n}$, and we take the expectation over the time-sharing variable $Q$ as
well when bounding the expectation of the average error probability. Thus, we
omit the proof of the below corollary.

\begin{corollary}
\label{cor:two-sender-QSD} Suppose that the rates $R_{1}$ and $R_{2}$ satisfy
the following inequalities:%
\begin{align}
R_{1} &  \leq I\left(  X;B|YQ\right)  ,\\
R_{2} &  \leq I\left(  Y;B|XQ\right)  _{\rho},\\
R_{1}+R_{2} &  \leq I\left(  XY;B|Q\right)  _{\rho},
\end{align}
where the entropies are with respect to a state of the following form:
\begin{multline*}
\rho^{QXYB}\equiv\sum_{x,y,q}p_{Q}(q)\,p_{X|Q}\left(  x|q\right)
\,p_{Y|Q}\left(  y|q\right) \\[-1mm]
\left\vert q\right\rangle \left\langle q\right\vert ^{Q}\otimes
\left\vert x\right\rangle \left\langle x\right\vert ^{X}\otimes\left\vert
y\right\rangle \left\langle y\right\vert ^{Y}\otimes\rho_{x,y}^{B}.
\end{multline*}
Then, if the codebooks for Senders~1 and 2 are chosen as described above, 
there exists a 
corresponding simultaneous decoding POVM $\left\{  \Lambda_{l,m}\right\}  $
such that the expectation of the average probability of error is bounded above
by $\epsilon$ for all $\epsilon>0$ and sufficiently large $n$.
\end{corollary}

\subsubsection{Conjecture for Three-Sender Quantum Simultaneous Decoding}

We now state our conjecture regarding the existence of a quantum simultaneous
decoder for a quantum multiple access channel with three classical inputs.
We state the conjecture for a
three-sender quantum multiple access channel because this form is the one
required for the proof of the Han-Kobayashi achievable rate region \cite{HK81}.

\begin{conjecture}
[Existence of a Three-Sender Quantum Simultaneous Decoder]\label{conj:sim-dec}Let
$x,y,z\rightarrow\rho_{x,y,z}$ be a cccq quantum multiple access channel,
where Sender~1 has access to the $x$ input, Sender~2 has access to the $y$
input, and Sender~3 has access to the $z$ input. 
Let $p_X, p_Y$ and $p_Z$ be distributions on the inputs. Define the following random code:
let $\{X^n(k)\}_{k \in \{1, \dots, K\}}$ be independent random variables distributed according to the
product distribution $p_{X^n}$ and similarly and independently let $\{Y^n(l)\}_{l \in \{1, \dots, L\}}$ and $\{Z^n(m)\}_{m \in \{1, \dots, M\}}$ be independent random variables distributed according to product distributions $p_{Y^n}$ and $p_{Z^n}$. The rates of communication are $R_{1}=\frac{1}%
{n}\log_{2}\left(  K\right) +\delta $, $R_{2}=\frac{1}{n}\log_{2}\left(
L\right)  +\delta$, and $R_{3}=\frac{1}{n}\log_{2}\left(  M\right)  +\delta$,
respectively, where $\delta > 0$. Suppose that these rates obey the following inequalities:%
\begin{align*}
R_{1}  &  \leq I\left(  X;B|YZ\right)  _{\rho},\\
R_{2}  &  \leq I\left(  Y;B|XZ\right)  _{\rho},\\
R_{3}  &  \leq I\left(  Z;B|XY\right)  _{\rho},\\
R_{1}+R_{2}  &  \leq I\left(  XY;B|Z\right)  _{\rho},\\
R_{1}+R_{3}  &  \leq I\left(  XZ;B|Y\right)  _{\rho},\\
R_{2}+R_{3}  &  \leq I\left(  YZ;B|X\right)  _{\rho},\\
R_{1}+R_{2}+R_{3}  &  \leq I\left(  XYZ;B\right)  _{\rho},
\end{align*}
where the Holevo information quantities are with respect to the following
classical-quantum state:%
\begin{multline}
\label{eq:in-out-3mac}
\rho^{XYZB} \equiv  \sum_{x,y,z}p_{X}\!\left(  x\right)  p_{Y}\!\left(  y\right)
p_{Z}\!\left(  z\right)  \left\vert x\right\rangle \left\langle x\right\vert
^{X}\otimes\left\vert y\right\rangle \left\langle y\right\vert ^{Y} \\
\otimes\left\vert z\right\rangle \left\langle z\right\vert ^{Z}\otimes
\rho_{x,y,z}^{B}.
\end{multline}
Then there exists a decoding POVM $\left\{  \Lambda_{l,m,k}\right\}_{l,m,k}  $ such that the
expectation of the average probability of error is bounded above by $\epsilon$
for all $\epsilon>0$ and sufficiently large $n$:%
\[
\mathbb{E}\left\{  \frac{1}{KLM}\sum_{k,l,m}%
\text{Tr}\left\{  \left(  I-\Lambda_{k,l,m}\right)  \rho_{X^{n}\left(
k\right)  ,Y^{n}\left(  l\right)  ,Z^{n}\left(  m\right)  }\right\}  \right\}
\leq\epsilon,
\]
where the expectation is with respect to $X^{n}$, $Y^{n}$, and $Z^{n}$.

\end{conjecture}

The importance of this conjecture stems not only from the fact that a proof of it would
be helpful in achieving a ``quantized'' version of the Han-Kobayashi achievable rate region,
but also because such a proof might more broadly be helpful for ``quantizing'' other results in network
classical information theory. Indeed, many coding theorems in network classical information
theory exploit a simultaneous decoding approach (sometimes known as jointly typical decoding)
\cite{el2010lecture}. Also, Dutil and Hayden have recently put forward a related conjecture
known as the ``multiparty typicality'' conjecture \cite{D11}, and it is likely that a proof of
Conjecture~\ref{conj:sim-dec} could aid in producing a proof of the multiparty typicality conjecture
or vice versa.

\subsubsection{Special Cases of the Conjecture}

\label{sec:min-entropy}We now offer two theorems that are variations of the above
conjecture that do hold for three-sender multiple access channels.
The first is a special case in which we assume that certain averaged output
states commute, and the second is one in which certain bounds contain min-entropies. 
It seems likely that an eventual proof of
Conjecture~\ref{conj:sim-dec}, should one be found,
will involve steps similar to those presented below,
albeit with some crucial additional ideas.

\paragraph{Commuting Case}

We prove a special case of Conjecture~\ref{conj:sim-dec} in which we assume that certain
averaged output states commute. First, let us define the following states%
\begin{align*}
\rho_{x,z}  &  \equiv\sum_{y}p_{Y}\left(  y\right)  \ \rho_{x,y,z}%
,\\
 \rho_{y,z}& \equiv\sum_{x}p_{X}\left(  x\right)  \ \rho_{x,y,z}%
,\\ 
\rho_{x,y}& \equiv\sum_{z}p_{Z}\left(  z\right)  \ \rho_{x,y,z},\\
\rho_{x}  &  \equiv\sum_{z}p_{Z}\left(  z\right)  \ \rho_{x,z},\\
\rho_{y} & \equiv\sum_{x}p_{X}\left(  x\right)  \ \rho_{x,y},\\
\rho_{z} & \equiv \sum_{y}p_{Y}\left(  y\right)  \ \rho_{y,z},\\
\rho &  \equiv\sum_{x,y,z}p_{X}\left(  x\right)  \ p_{Y}\left(  y\right)
\ p_{Z}\left(  z\right)  \ \rho_{x,y,z}.
\end{align*}

\begin{theorem}
[Averaged State Commuting Case]Consider the same setup as in Conjecture~\ref{conj:sim-dec}, with
the additional assumption that certain averaged states commute:$\ \left[
\rho_{x,z},\rho_{y,z}\right]  =\left[  \rho_{x,y},\rho_{y,z}\right]  =\left[
\rho_{x,y},\rho_{x,z}\right]  =0$ for all $x\in\mathcal{X}$, $y\in\mathcal{Y}%
$, and $z\in\mathcal{Z}$. Then there exists a quantum simultaneous decoder in
the sense described in Conjecture~\ref{conj:sim-dec}.
\end{theorem}

\begin{proof}
The proof exploits some ideas from Theorem~\ref{thm:two-sender-QSD}. Thus, we merely describe the key
points of the proof.

We randomly and independently choose codewords for the three senders according
to the respective product distributions $p_{X^{n}}\left(  x^{n}\right)  $,
$p_{Y^{n}}\left(  y^{n}\right)  $, and $p_{Z^{n}}\left(  z^{n}\right)  $. We
define the detection POVM\ to be of the following form:%
\begin{multline}
\Lambda_{k,l,m}\equiv \\ \left(  \sum_{k^{\prime},l^{\prime},m^{\prime}}%
\Pi_{k^{\prime},l^{\prime},m^{\prime}}^{\prime}\right)  ^{-1/2}\Pi
_{k,l,m}^{\prime}\left(  \sum_{k^{\prime},l^{\prime},m^{\prime}}\Pi
_{k^{\prime},l^{\prime},m^{\prime}}^{\prime}\right)  ^{-1/2}, \label{eq:sqrt-POVM}
\end{multline}
where%
\begin{align*}
\Pi_{k,l,m}^{\prime} &  \equiv M_{x^{n}\left(  k\right)  ,y^{n}\left(
l\right)  ,z^{n}\left(  m\right)  }^{\dag}M_{x^{n}\left(  k\right)
,y^{n}\left(  l\right)  ,z^{n}\left(  m\right)  },\\
M_{x^{n},y^{n},z^{n}} &  \equiv\Pi_{x^{n},y^{n},z^{n}}\ \Pi_{x^{n},y^{n}}%
\ \Pi_{x^{n},z^{n}}\ \Pi_{y^{n},z^{n}} \times \\ & \,\,\,\,\,\,\,\,\,\,\,\, \Pi_{x^{n}}\ \Pi_{y^{n}}\ \Pi_{z^{n}%
}\ \Pi,
\end{align*}
and each of the above projectors are conditionally typical projectors defined
with a similar shorthand from the proof of Theorem~\ref{thm:two-sender-QSD}. Observe that all of the
conditionally typical projectors $\Pi_{x^{n},y^{n}}$, $\Pi_{x^{n},z^{n}}$,
$\Pi_{y^{n},z^{n}}$, $\Pi_{x^{n}}$, $\Pi_{y^{n}}$, $\Pi_{z^{n}}$, and $\Pi$
are mutually commuting from the assumption of the theorem. We analyze the
expectation of the average error probability, and due to the symmetry of the
code construction, it suffices to analyze this error probability for the first
message triple $\left(  1,1,1\right)  $:%
\[
\mathbb{E}_{X^{n},Y^{n},Z^{n}}\left\{  \text{Tr}\left\{  \left(
I-\Lambda_{1,1,1}\right)  \ \rho_{X^{n}\left(  1\right)  ,Y^{n}\left(
1\right)  ,Z^{n}\left(  1\right)  }\right\}  \right\}  .
\]
Our first move is to \textquotedblleft unravel\textquotedblright\ the operator
$I-\Lambda_{1,1,1}$ by means of the Hayashi-Nagaoka operator inequality, so
that%
\[
I-\Lambda_{1,1,1}\leq2 \left(  I-\Pi_{1,1,1}^{\prime}\right)  +\ \ \ \ 4\!\!\!\!\!\!\!\!\!\!\!\sum_{\left(
k,l,m\right)  \neq\left(  1,1,1\right)  }\!\!\!\!\!\!\Pi_{k,l,m}^{\prime}.
\]
The first error with the operator $I-\Pi_{1,1,1}^{\prime}$ under the trace can
be bounded from above by some $f\left(  \epsilon\right)  $ where
$\lim_{\epsilon\rightarrow0}f\left(  \epsilon\right)  =0$, by employing the
trace inequality in (\ref{eq:trace-inequality}) and the Gentle Operator Lemma for ensembles. 
We can expand the triply-indexed sum for the second
error into seven different types of errors. We delineate the different errors in the following table:
\begin{equation}%
\begin{tabular}
[c]{c|c|c}\hline\hline
$k$ & $l$ & $m$\\\hline\hline
$\ast$ & $1$ & $1$\\
$1$ & $\ast$ & $1$\\
$1$ & $1$ & $\ast$\\
$\ast$ & $\ast$ & $1$\\
$\ast$ & $1$ & $\ast$\\
$1$ & $\ast$ & $\ast$\\
$\ast$ & $\ast$ & $\ast$\\\hline\hline
\end{tabular}
\ ,\label{eq:seven-errors}%
\end{equation}
where $\ast$ denotes some message other than the first one (implying an
incorrect decoding). Each of these we can bound by averaging over the state
$\rho_{X^{n}\left(  1\right)  ,Y^{n}\left(
1\right)  ,Z^{n}\left(  1\right)  }$
and commuting the appropriate projector to be closest to the state. For
example, consider the first error term. We have that $X^{n}\left(  k\right)  $
and $X^{n}\left(  1\right)  $ are independent. Bring the expectation over
$X^{n}$ inside of the trace and average over the state $\rho_{X^{n}\left(
1\right)  ,Y^{n}\left(  1\right)  ,Z^{n}\left(  1\right)  }$ to get
$\rho_{Y^{n}\left(  1\right)  ,Z^{n}\left(  1\right)  }$. Commute $\Pi
_{Y^{n}\left(  1\right)  ,Z^{n}\left(  1\right)  }$ to be closest to the state
on both sides and exploit the operator inequality $\Pi_{y^{n},z^{n}}%
\ \rho_{y^{n},z^{n}}\ \Pi_{y^{n},z^{n}}\leq2^{-n\left[  H\left(  B|YZ\right)
-\delta\right]  }\ \Pi_{y^{n},z^{n}}$. After a few steps, we end up with the
bound $2^{-n\left[  I\left(  X;B|YZ\right)  -2\delta\right]  }\ K$. The other six
bounds proceed in a similar fashion, demonstrating that Conjecture~\ref{conj:sim-dec} holds true
for this special case.
\end{proof}

\paragraph{Min-Entropy Case}A simple modification of the proof of
Theorem~\ref{thm:two-sender-QSD}\ allows us to achieve rates expressible in
terms of min-entropies \cite{R60,R05} for arbitrary quantum channels.
The min-entropy $H_{\min}\left(  B\right)  _{\rho}$ of a quantum
state $\rho^{B}$ is equal to the negative logarithm of its maximal eigenvalue:%
\[
H_{\min}\left(  B\right)  _{\rho}\equiv-\log\left(  \inf_{\lambda\in
\mathbb{R}}\left\{  \lambda:\rho\leq\lambda I\right\}  \right)  ,
\]
and the conditional min-entropy of a classical-quantum state $\rho^{XB}%
\equiv\sum_{x}p_{X}\left(  x\right)  \left\vert x\right\rangle \left\langle
x\right\vert ^{X}\otimes\rho_{x}^{B}$ with classical system $X$ and quantum
system $B$ is as follows~\cite{R05}:%
\[
H_{\min}\left(  B|X\right)  _{\rho}\equiv\inf_{x\in\mathcal{X}}H_{\min}\left(
B\right)  _{\rho_{x}}.
\]
This definition of conditional min-entropy, where the conditioning system is
classical, implies the following operator inequality:%
\begin{equation}
\forall x\ \ \ \rho_{x}^B \leq2^{-H_{\min}\left(  B|X\right)  _{\rho}}I^B.
\label{eq:cond-min-entropy-bound}%
\end{equation}

The following theorem gives an achievable rate region for a three-sender quantum
simultaneous decoder. The entropy differences in (\ref{eq:min-ent-first}-\ref{eq:min-ent-2})
and (\ref{eq:min-ent-3}-\ref{eq:min-ent-4}) of the following theorem
may not necessarily be positive for all states because the conditional quantum
min-entropy can be less than the conditional von Neumann entropy.
Nevertheless, there are some states for which these rates are positive,
and Example~\ref{ex:min-entropy} gives a channel for which the min-entropy rates are
equivalent to the von Neumann entropy rates.

\begin{theorem}
[Min-Entropy Case]\label{thm:min-ent}Consider the same setup as in Conjecture~\ref{conj:sim-dec}. There exists a
quantum simultaneous decoder in the sense described in Conjecture~\ref{conj:sim-dec} that
achieves the following rate region:%
\begin{align}
R_{1} &  \leq H_{\min}\left(  B|ZY\right)  -H\left(  B|XYZ\right)  , \label{eq:min-ent-first}\\
R_{2} &  \leq H_{\min}\left(  B|XZ\right)  -H\left(  B|XYZ\right)  , \label{eq:min-ent-2}\\
R_{3} &  \leq I\left(  Z;B|XY\right)  ,\\
R_{1}+R_{2} &  \leq H_{\min}\left(  B|Z\right)  -H\left(  B|XYZ\right)  ,\label{eq:min-ent-3}\\
R_{2}+R_{3} &  \leq H_{\min}\left(  B|X\right)  -H\left(  B|XYZ\right)  ,\label{eq:min-ent-4}\\
R_{1}+R_{3} &  \leq I\left(  XZ;B|Y\right)  ,\\
R_{1}+R_{2}+R_{3} &  \leq I\left(  XYZ;B\right)  \label{eq:min-ent-last}.
\end{align}
Other variations of the above achievable rate region are possible by permuting
the variables $X$, $Y$, and $Z$ in the above expressions.
\end{theorem}

\begin{proof}
The main idea for this proof is to exploit a decoding POVM\ of the form in (\ref{eq:sqrt-POVM}),
with $\Pi_{k,l,m}^{\prime}$ chosen to be as follows:%
\begin{multline}
\Pi_{k,l,m}^{\prime}=\Pi\ \Pi_{y^{n}\left(  l\right)  }\ \Pi_{x^{n}\left(
k\right)  ,y^{n}\left(  l\right)  }\ \Pi_{x^{n}\left(  k\right)  ,y^{n}\left(
l\right)  ,z^{n}\left(  m\right)  } \times \\\,\,\,\,\,\,\,\,
 \Pi_{x^{n}\left(  k\right)  ,y^{n}\left(
l\right)  }\ \Pi_{y^{n}\left(  l\right)  }\ \Pi.
\end{multline}
We can bound the expectation of the average error probability again by
exploiting the Hayashi-Nagaoka operator inequality. After doing so, the first
error with the operator $I-\Pi_{1,1,1}^{\prime}$ under the trace can be
bounded from above by some $f\left(  \epsilon\right)  $ where $\lim
_{\epsilon\rightarrow0}f\left(  \epsilon\right)  =0$, by employing the trace
inequality in (\ref{eq:trace-inequality}) and the Gentle Operator Lemma for ensembles. The second error
again breaks into the seven errors of the form in (\ref{eq:seven-errors}). We
discuss below how to handle each of these errors:

\begin{enumerate}
\item $X^{n}\left(  k\right)  $ and $X^{n}\left(  1\right)  $ are independent.
Bring the expectation over $X^{n}$ inside of the trace and average over the
state $\rho_{X^{n}\left(  1\right)  ,Y^{n}\left(  1\right)  ,Z^{n}\left(
1\right)  }$ to get $\rho_{Y^{n}\left(  1\right)  ,Z^{n}\left(  1\right)  }$.
The state $\rho_{Y^{n}\left(  1\right)  ,Z^{n}\left(  1\right)  }$ is bounded
from above by $2^{-nH_{\min}\left(  B|YZ\right)  }$ and proceed to upper bound
this error by $2^{-n\left[  H_{\min}\left(  B|ZY\right)  -H\left(
B|XYZ\right)  \right]  }\ K$.

\item $Y^{n}\left(  l\right)  $ and $Y^{n}\left(  1\right)  $ are independent.
Bring the expectation over $Y^{n}$ inside of the trace and average over the
state $\rho_{X^{n}\left(  1\right)  ,Y^{n}\left(  1\right)  ,Z^{n}\left(
1\right)  }$ to get $\rho_{X^{n}\left(  1\right)  ,Z^{n}\left(  1\right)  }$.
The state $\rho_{X^{n}\left(  1\right)  ,Z^{n}\left(  1\right)  }$ is bounded
from above by $2^{-nH_{\min}\left(  B|XZ\right)  }$ and proceed to upper bound
this error by $2^{-n\left[  H_{\min}\left(  B|XZ\right)  -H\left(
B|XYZ\right)  \right]  }\ L$.

\item $Z^{n}\left(  m\right)  $ and $Z^{n}\left(  1\right)  $ are independent.
Exploit the operator inequality $\Pi_{x^{n},y^{n},z^{n}}\leq2^{n\left[
H\left(  B|XYZ\right)  +\delta\right]  }\ \rho_{x^{n},y^{n},z^{n}}$, bring the
expectation over $Z^{n}$ inside of the trace and average over the state
$\rho_{X^{n}\left(  1\right)  ,Y^{n}\left(  1\right)  ,Z^{n}\left(  M\right)
}$ to get $\rho_{X^{n}\left(  1\right)  ,Y^{n}\left(  1\right)  }$. Exploit
the operator inequality $\Pi_{x^{n},y^{n}}\ \rho_{x^{n},y^{n}}\ \Pi
_{x^{n},y^{n}}\leq2^{-n\left[  H\left(  B|XY\right)  -\delta\right]  }%
\ \Pi_{x^{n},y^{n}}$. We can then upper bound this error by $2^{-n\left[
I\left(  Z;B|XY\right)  -2\delta\right]  }\ M$.

\item $X^{n}\left(  k\right)  $ and $X^{n}\left(  1\right)  $ are independent,
and so are $Y^{n}\left(  l\right)  $ and $Y^{n}\left(  1\right)  $. Bring the
expectations over $X^{n}$ and $Y^{n}$ inside of the trace and average over the
state $\rho_{X^{n}\left(  1\right)  ,Y^{n}\left(  1\right)  ,Z^{n}\left(
1\right)  }$ to get $\rho_{Z^{n}\left(  1\right)  }$. The state $\rho
_{Z^{n}\left(  1\right)  }$ is bounded from above by $2^{-nH_{\min}\left(
B|Z\right)  }$ and proceed to upper bound this error by $2^{-n\left[  H_{\min
}\left(  B|Z\right)  -H\left(  B|XYZ\right)  \right]  }\ KL$.

\item $Y^{n}\left(  l\right)  $ and $Y^{n}\left(  1\right)  $ are independent,
and so are $Z^{n}\left(  m\right)  $ and $Z^{n}\left(  1\right)  $. Bring the
expectations over $Y^{n}$ and $Z^{n}$ inside of the trace and average over the
state $\rho_{X^{n}\left(  1\right)  ,Y^{n}\left(  1\right)  ,Z^{n}\left(
1\right)  }$ to get $\rho_{X^{n}\left(  1\right)  }$. The state $\rho
_{X^{n}\left(  1\right)  }$ is bounded from above by $2^{-nH_{\min}\left(
B|X\right)  }$ and proceed to upper bound this error by $2^{-n\left[  H_{\min
}\left(  B|X\right)  -H\left(  B|XYZ\right)  \right]  }\ LM$.

\item $X^{n}\left(  k\right)  $ and $X^{n}\left(  1\right)  $ are independent,
and so are $Z^{n}\left(  m\right)  $ and $Z^{n}\left(  1\right)  $. Exploit
the operator inequality $\Pi_{x^{n},y^{n},z^{n}}\leq2^{n\left[  H\left(
B|XYZ\right)  +\delta\right]  }\ \rho_{x^{n},y^{n},z^{n}}$, bring the
expectation over $Z^{n}$ inside of the trace and average over the state
$\rho_{X^{n}\left(  1\right)  ,Y^{n}\left(  1\right)  ,Z^{n}\left(  M\right)
}$ to get $\rho_{X^{n}\left(  1\right)  ,Y^{n}\left(  1\right)  }$. Exploit
the operator inequality $\Pi_{x^{n},y^{n}}\ \rho_{x^{n},y^{n}}\ \Pi
_{x^{n},y^{n}}\leq\rho_{x^{n},y^{n}}$. Bring the expectation over $X^{n}$
inside of the trace and average over the state $\rho_{X^{n}\left(  1\right)
,Y^{n}\left(  1\right)  }$ to get $\rho_{Y^{n}\left(  1\right)  }$. Exploit
the operator inequality $\Pi_{y^{n}}\ \rho_{y^{n}}\ \Pi_{y^{n}}\leq
2^{-n\left[  H\left(  B|Y\right)  -\delta\right]  }\ \Pi_{y^{n}}$. We can then
upper bound this error by $2^{-n\left[  I\left(  XZ;B|Y\right)  -2\delta
\right]  }\ KM$.

\item All variables are independent. Bring the expectations over $X^{n}$,
$Y^{n}$, and $Z^{n}$ inside of the trace and average over the state
$\rho_{X^{n}\left(  1\right)  ,Y^{n}\left(  1\right)  ,Z^{n}\left(  1\right)
}$ to get $\rho^{\otimes n}$. Exploit the operator inequality $\Pi
\ \rho^{\otimes n}\ \Pi\leq2^{-n\left[  H\left(  B\right)  -\delta\right]
}\ \Pi$ and proceed to upper bound this error by $2^{-n\left[  I\left(
XYZ;B\right)  -2\delta\right]  }\ KLM$.
\end{enumerate}
\end{proof}

\begin{example}
\label{ex:min-entropy} We now provide an example of a cccq quantum multiple
access channel for which a quantum simultaneous decoder can achieve its
capacity region. We show that the min-entropy rates in (\ref{eq:min-ent-first}%
-\ref{eq:min-ent-last}) of Theorem~\ref{thm:min-ent} are equal to the von
Neumann entropy rates from Conjecture~\ref{conj:sim-dec}. By Winter's results
in Ref.~\cite{winter2001capacity} for a cccq multiple access channel, this
implies that the min-entropy rate region is equivalent to the capacity region
for this particular channel. Consider a channel that takes three bits $x$,
$y$, and $z$ as input and outputs one of the four \textquotedblleft
BB84\textquotedblright\ states:
\begin{align*}
000 &  \rightarrow\left\vert 0\right\rangle ,\ \ \ \ \ 001\rightarrow
\left\vert +\right\rangle ,\ \ \ \ \ 010\rightarrow\left\vert 1\right\rangle
,\ \ \ \ \ 011\rightarrow\left\vert -\right\rangle ,\\
100 &  \rightarrow\left\vert 1\right\rangle ,\ \ \ \ \ 101\rightarrow
\left\vert -\right\rangle ,\ \ \ \ \ 110\rightarrow\left\vert 0\right\rangle
,\ \ \ \ \ 111\rightarrow\left\vert +\right\rangle .
\end{align*}
A classical-quantum state on which we evaluate information quantities is
\begin{multline*}
\rho^{XYZB}\equiv\sum_{x,y,z=0}^{1}p_{X}\left(  x\right)  p_{Y}\left(
y\right)  p_{Z}\left(  z\right)  \left\vert x\right\rangle \left\langle
x\right\vert ^{X}\otimes\left\vert y\right\rangle \left\langle y\right\vert
^{Y}\otimes\\
\left\vert z\right\rangle \left\langle z\right\vert ^{Z}\otimes\psi
_{x,y,z}^{B},
\end{multline*}
where $\psi_{x,y,z}^{B}$ is one of $\left\vert 0\right\rangle $, $\left\vert
1\right\rangle $, $\left\vert +\right\rangle $, or $\left\vert -\right\rangle
$ depending on the choice of the bits $x$, $y$, and $z$. The conditional
entropy $H\left(  B|XYZ\right)  _{\rho}$ vanishes for this state because the
state is pure when conditioned on the classical registers $X$, $Y$, and $Z$.
So it is only necessary to compare $H_{\min}\left(  B|ZY\right)  $ with
$H\left(  B|ZY\right)  $, $H_{\min}\left(  B|XZ\right)  $ with $H\left(
B|XZ\right)  $, $H_{\min}\left(  B|Z\right)  $ with $H\left(  B|Z\right)  $,
and $H_{\min}\left(  B|X\right)  $ with $H\left(  B|X\right)  $. We choose
$p_{X}\left(  x\right)  $, $p_{Y}\left(  y\right)  $, and $p_{Z}\left(
z\right)  $ to be the uniform distribution. This gives the following reduced
state on $Z$, $Y$, and $B$:
\begin{multline*}
\frac{1}{4}\left\vert 00\right\rangle \left\langle 00\right\vert ^{ZY}%
\otimes\frac{1}{2}\left(  \left\vert 0\right\rangle \left\langle 0\right\vert
^{B}+\left\vert 1\right\rangle \left\langle 1\right\vert ^{B}\right)  \\
+\frac{1}{4}\left\vert 01\right\rangle \left\langle 01\right\vert ^{ZY}%
\otimes\frac{1}{2}\left(  \left\vert +\right\rangle \left\langle +\right\vert
^{B}+\left\vert -\right\rangle \left\langle -\right\vert ^{B}\right)  \\
+\frac{1}{4}\left\vert 10\right\rangle \left\langle 10\right\vert ^{ZY}%
\otimes\frac{1}{2}\left(  \left\vert 0\right\rangle \left\langle 0\right\vert
^{B}+\left\vert 1\right\rangle \left\langle 1\right\vert ^{B}\right)  \\
+\frac{1}{4}\left\vert 11\right\rangle \left\langle 11\right\vert ^{ZY}%
\otimes\frac{1}{2}\left(  \left\vert +\right\rangle \left\langle +\right\vert
^{B}+\left\vert -\right\rangle \left\langle -\right\vert ^{B}\right)  ,
\end{multline*}
for which it is straightforward to show that certain entropies take their
maximal value of one bit: $H_{\min}\left(  B|ZY\right)  =H\left(  B|ZY\right)
=1$ and $H_{\min}\left(  B|Z\right)  =H\left(  B|Z\right)  =1$. We also have
the following reduced state on $X$, $Z$, and $B$:%
\begin{multline*}
\frac{1}{4}\left\vert 00\right\rangle \left\langle 00\right\vert ^{XZ}%
\otimes\frac{1}{2}\left(  \left\vert 0\right\rangle \left\langle 0\right\vert
^{B}+\left\vert 1\right\rangle \left\langle 1\right\vert ^{B}\right)  \\
+\frac{1}{4}\left\vert 01\right\rangle \left\langle 01\right\vert ^{XZ}%
\otimes\frac{1}{2}\left(  \left\vert +\right\rangle \left\langle +\right\vert
^{B}+\left\vert -\right\rangle \left\langle -\right\vert ^{B}\right)  \\
+\frac{1}{4}\left\vert 10\right\rangle \left\langle 10\right\vert ^{XZ}%
\otimes\frac{1}{2}\left(  \left\vert 0\right\rangle \left\langle 0\right\vert
^{B}+\left\vert 1\right\rangle \left\langle 1\right\vert ^{B}\right)  \\
+\frac{1}{4}\left\vert 11\right\rangle \left\langle 11\right\vert ^{XZ}%
\otimes\frac{1}{2}\left(  \left\vert +\right\rangle \left\langle +\right\vert
^{B}+\left\vert -\right\rangle \left\langle -\right\vert ^{B}\right)  ,
\end{multline*}
for which the other entropies take their maximal value of one bit: $H_{\min
}\left(  B|XZ\right)  =H\left(  B|XZ\right)  =1$ and $H_{\min}\left(
B|X\right)  =H\left(  B|X\right)  =1$. Furthermore, we can show that the
conditional entropy $H\left(  B|XY\right)  _{\rho}$ takes it maximum value of
$H_{2}\left(  \cos^{2}\left(  \pi/8\right)  \right)  $ when $p_{X}\left(
x\right)  $ and $p_{Y}\left(  y\right)  $ are uniform (where $H_{2}\left(
p\right)  \equiv-p\log_{2}p-\left(  1-p\right)  \log_{2}\left(  1-p\right)
$). Thus, the region achievable with min-entropies in (\ref{eq:min-ent-first}%
-\ref{eq:min-ent-last}) of Theorem~\ref{thm:min-ent} is equivalent to the
capacity region for this channel:
\begin{align*}
R_{1} &  \leq1,\\
R_{2} &  \leq1,\\
R_{3} &  \leq H_{2}\left(  \cos^{2}\left(  \pi/8\right)  \right)  ,\\
R_{1}+R_{2} &  \leq1,\\
R_{2}+R_{3} &  \leq1,\\
R_{1}+R_{3} &  \leq1,\\
R_{1}+R_{2}+R_{3} &  \leq1.
\end{align*}

\end{example}

\subsubsection{Other Attempts at Proving Conjecture~\ref{conj:sim-dec}}

We have attempted to prove Conjecture~\ref{conj:sim-dec}\ in many different
ways, 
and this section briefly summarizes these attempts. We again mention that
our quantum simultaneous decoding conjecture seems related to the multiparty
typicality conjecture from Ref.~\cite{D11}.

We have attempted to prove Conjecture~\ref{conj:sim-dec} by exploiting
the asymmetric hypothesis testing techniques from
Refs.~\cite{WR10,mosonyi:072104}. The problem with these approaches in the
multiple access setting is that the POVM\ selected in the operational
definitions of the quantum relative entropy is optimal for one type of error
in (\ref{eq:three-terms}), but it is not necessarily optimal for the other two
types of errors. The hypothesis testing approaches from
Refs.~\cite{BP10,BDKSSS04}\ also do not appear to be of much help for our
goals here because they involve an infimum over the choice of the second state
in the quantum relative entropy.

Another attempt is to improve the achievable rate region of
Theorem~\ref{thm:min-ent}, by replacing
min-entropies with \emph{smooth} min-entropies \cite{R05}. In fact, the smooth min-entropy
is known to approach the von Neumann entropy in the case of a large number of independent 
and identically distributed random variables~\cite{R05, TRR09}. 
To prove the conjecture, it would be sufficient to find a state $\tilde{\rho}^{X^nY^nZ^nB^n}$ that is close to $\rho^{X^nY^nZ^nB^n}$---which corresponds to $n$ independent copies of the state $\rho^{XYZB}$ in \eqref{eq:in-out-3mac}---that simultaneously satisfies $\entHmin(B|ZY)_{\tilde{\rho}} \geq \entHmin^{\e}(B|ZY)_{\rho}$, $\entHmin(B|XZ)_{\tilde{\rho}} \geq \entHmin^{\e}(B|XZ)_{\rho}$, $\entHmin(B|Z)_{\tilde{\rho}} \geq \entHmin^{\e}(B|Z)_{\rho}$, $\entHmin(B|X)_{\tilde{\rho}} \geq \entHmin^{\e}(B|X)_{\rho}$. 
Here, $\entHmin^{\e}(B|X)_{\rho}$ refers to the $\e$-smooth min-entropy, which is the maximum of $\entHmin(B|X)_{\rho'}$ over all states $\rho'$ on $XB$ that are $\e$-close to $\rho$; see \cite{R05} for a precise definition. In the proof of Theorem \ref{thm:min-ent}, we would replace the output of the channel $\rho$ by $\tilde{\rho}$ before applying the Hayashi-Nagaoka operator inequality and the min-entropy terms would approach the von Neumann entropy terms we are looking for.

\bigskip 
\section{The Quantum Interference Channel}

\label{sec:QIC}This section contains some of the main results of this
paper,\ the inner and outer bounds on the capacity of a \textit{ccqq} quantum
interference channel of the following form:%
\begin{equation}
x_{1},x_{2}\rightarrow\rho_{x_{1},x_{2}}^{B_{1}B_{2}}, \label{eq:ccqq-int}%
\end{equation}
where Sender~1 has access to the classical $x_{1}$ input, Sender~2 has access
to the classical $x_{2}$ input, Receiver~1 has access to the $B_{1}$ quantum
system, and Receiver~2 has access to the $B_{2}$ quantum system. The first
inner bound that we prove is similar to the result of Carleial for
\textquotedblleft very strong\textquotedblright\ interference. We then prove
a quantum simultaneous decoding inner bound and give the capacity of the channel
whenever it exhibits ``strong'' interference.
The main inner
bound is the Han-Kobayashi achievable rate region with Shannon information
quantities replaced by Holevo information quantities, and this inner bound
relies on Conjecture~\ref{conj:sim-dec} for its proof. The outer bound in
Section~\ref{sec:outer-bound}\ is similar to an outer bound in the
classical case due to Sato~\cite{Sato77}.

\subsection{Inner Bounds}

As mentioned earlier, the interference channel 
naturally induces two multiple access channels with the same senders. Thus, 
one possible coding strategy for the interference channel is to build a codebook for each 
multiple access channel that is decodable for \emph{both} receivers. In fact, most---if not all---known coding
strategies for the interference channel are based on this idea. 
It is important to say here that we have to use the \emph{same} codebook for both 
multiple access channels. For this reason, using the \emph{existence} of 
good codes achieving all tuples in the capacity region is not sufficient.

\subsubsection{Very Strong Interference}

\label{sec:very-strong}A setting for which we can determine the capacity of a
ccqq interference channel is the setting of \textquotedblleft very
strong\textquotedblright\ interference (see page 6-11 of
Ref.~\cite{el2010lecture}). The conditions for \textquotedblleft very
strong\textquotedblright\ interference are that the following information
inequalities should hold for all distributions $p_{X_{1}}\left(  x_{1}\right)
$ and $p_{X_{2}}\left(  x_{2}\right)  $:%
\begin{align}
I\left(  X_{1};B_{1}|X_{2}\right)  _{\rho}  &  \leq I\left(  X_{1}%
;B_{2}\right)  _{\rho},\label{eq:VSI-1}\\
I\left(  X_{2};B_{2}|X_{1}\right)  _{\rho}  &  \leq I\left(  X_{2}%
;B_{1}\right)  _{\rho}, \label{eq:VSI-2}%
\end{align}
where $\rho^{X_{1}X_{2}B_{1}B_{2}}$ is a state of the following form:%
\begin{multline}
\rho^{X_{1}X_{2}B_{1}B_{2}}\equiv\sum_{x_{1},x_{2}}p_{X_{1}}\left(
x_{1}\right)  p_{X_{2}}\left(  x_{2}\right)  \left\vert x_{1}\right\rangle
\left\langle x_{1}\right\vert ^{X_{1}}\otimes \\ \left\vert x_{2}\right\rangle
\left\langle x_{2}\right\vert ^{X_{2}}\otimes\rho_{x_{1},x_{2}}^{B_{1}B_{2}}. \label{eq:int-code-state}
\end{multline}
The information inequalities in (\ref{eq:VSI-1}-\ref{eq:VSI-2}) imply that the
interference is so strong that it is possible for each receiver to decode the
other sender's message before decoding the message intended for him. These conditions
are a generalization of Carleial's conditions for a classical Gaussian
interference channel~\cite{carleial1975case}.

\begin{theorem}[Very Strong Interference]
\label{thm:carleial}Let a ccqq quantum interference channel as in
(\ref{eq:ccqq-int}) be given, and suppose that it has \textquotedblleft very
strong\textquotedblright\ interference as in (\ref{eq:VSI-1}-\ref{eq:VSI-2}).
Then the channel's capacity region is the union of all rates $R_{1}$ and $R_{2}$ satisfying the below inequalities:
\begin{align*}
R_{1}  &  \leq I\left(  X_{1};B_{1}|X_{2} Q\right)  _{\rho},\\
R_{2}  &  \leq I\left(  X_{2};B_{2}|X_{1} Q\right)  _{\rho},
\end{align*}
where the union is over input distributions $p_Q(q) \, p_{X_{1} | Q}\left(x_{1}|q\right)\,
p_{X_{2}|Q}\left(  x_{2}|q\right)$.
\end{theorem}

\begin{proof}
Our proof technique is to apply Winter's successive decoder from
Lemma~\ref{thm:successive-decoder}, so that each receiver first decodes the
message of the other sender, followed by decoding the message of the partner
sender. More specifically, Senders~1 and 2 randomly choose a codebook of size
$L \approx 2^{nI\left(  X_{1};B_{1}|X_{2} Q\right)  }$ and $M \approx 2^{nI\left(  X_{2}%
;B_{2}|X_{1} Q\right)  }$, respectively. The choice of random code is such that
Receiver~1 can first decode the message $m$ because the message $m$ is
distinguishable whenever the message set size $M$ is less than $2^{nI\left(
X_{2};B_{1} | Q\right)  }$ and the very strong interference condition in
(\ref{eq:VSI-1}) guarantees that this holds. Receiver~1 then uses $X_{2}$ as
side information to decode message $l$ from Sender~1. Receiver~2 performs
similar steps by exploiting the very strong interference condition in
(\ref{eq:VSI-2}). The random choice of code guarantees that the expectation of
the average error probability is arbitrarily small, and this furthermore
guarantees the existence of a particular code with arbitrarily small average
error probability. The converse of this theorem follows by the same reasoning as
Carleial \cite{carleial1975case,el2010lecture}---the outer bound follows by considering that
the conditional mutual information rates in the statement of the theorem
are what they could achieve if Senders 1 and 2 maximize their rates individually.
\end{proof}

\begin{example}
\label{ex:theta-SWAP}
We now consider an example of a ccqq quantum interference channel with two
classical inputs and two quantum outputs:%
\begin{align}
00  &  \rightarrow\left\vert 00\right\rangle ^{B_{1}B_{2}}%
,\label{eq:patrick-example-1}\\
01  &  \rightarrow\cos\left(  \theta\right)  \left\vert 01\right\rangle
^{B_{1}B_{2}}+\sin\left(  \theta\right)  \left\vert 10\right\rangle
^{B_{1}B_{2}},\\
10  &  \rightarrow-\sin\left(  \theta\right)  \left\vert 01\right\rangle
^{B_{1}B_{2}}+\cos\left(  \theta\right)  \left\vert 10\right\rangle
^{B_{1}B_{2}},\\
11  &  \rightarrow\left\vert 11\right\rangle ^{B_{1}B_{2}}.
\label{eq:patrick-example-4}%
\end{align}
The first classical input is for Sender~1, and the second classical input is
for Sender~2. This transformation results if the two senders input one of the
four classical states $\left\{  \left\vert 00\right\rangle ,\left\vert
01\right\rangle ,\left\vert 10\right\rangle ,\left\vert 11\right\rangle
\right\}  $ to a \textquotedblleft$\theta$-SWAP\textquotedblright\ unitary
transformation that takes this computational basis to the output basis in
(\ref{eq:patrick-example-1}-\ref{eq:patrick-example-4}).

\begin{figure}
[ptb]
\begin{center}
\includegraphics[
natheight=3.806000in,
natwidth=5.460400in,
width=3.4411in
]%
{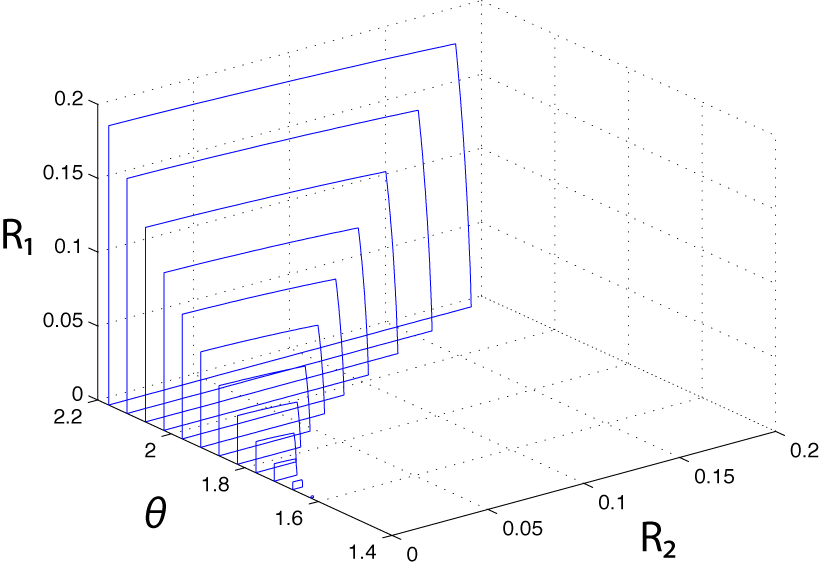}%
\caption{The capacity region of the \textquotedblleft$\theta$%
-SWAP\textquotedblright\ interference channel for various values of $\theta$
such that the channel exhibits \textquotedblleft very strong\textquotedblright%
\ interference. The capacity region is largest when $\theta$ gets closer to
2.18, and it vanishes when $\theta=\pi/2$ because the channel becomes a full
SWAP (at this point, Receiver~$i$ gets no information from Sender~$i$, where
$i\in\left\{  1,2\right\}  $). }%
\label{fig:patrick-example}%
\end{center}
\end{figure}
We would like to determine an interval for the parameter $\theta$ for which
the channel exhibits \textquotedblleft very strong\textquotedblright%
\ interference. In order to do so, we need to consider classical-quantum
states of the following form:%
\begin{multline}
\rho^{X_1 X_2 B_{1}B_{2}}\equiv\sum_{x_1,x_2=0}^{1}p_{X_1}\left(  x_1\right)  p_{X_2}\left(
x_2\right)  \left\vert x_1\right\rangle \left\langle x_1\right\vert ^{X_1}%
\otimes\\\left\vert x_2\right\rangle \left\langle x_2\right\vert ^{X_2}\otimes
\psi_{x_1,x_2}^{B_{1}B_{2}},\label{eq:patrick-example-cq-state}%
\end{multline}
where $\psi_{x_1,x_2}^{B_{1}B_{2}}$ is one of the pure output states in
(\ref{eq:patrick-example-1}-\ref{eq:patrick-example-4}). We should then check
whether the conditions in (\ref{eq:VSI-1}-\ref{eq:VSI-2}) hold for all
distributions $p_{X_1}\left(  x_1\right)  $ and $p_{X_2}\left(  x_2\right)  $. We can
equivalently express these conditions in terms of von Neumann entropies as
follows:%
\begin{align*}
H\left(  B_{1}|X_2\right)  _{\rho}-H\left(  B_{1}|X_1 X_2\right)  _{\rho} &  \leq
H\left(  B_{2}\right)  _{\rho}-H\left(  B_{2}|X_1\right)  _{\rho},\\
H\left(  B_{2}|X_1\right)  _{\rho}-H\left(  B_{2}| X_1 X_2\right)  _{\rho} &  \leq
H\left(  B_{1}\right)  _{\rho}-H\left(  B_{1}|X_2\right)  _{\rho},
\end{align*}
and thus, it suffices to calculate six entropies for states of the form in
(\ref{eq:patrick-example-cq-state}). After some straightforward calculations,
we find the results in (\ref{eq:entropies-1}-\ref{eq:entropies-last})%
\begin{figure*}
\begin{align}
H\left(  B_{1}|X_1 X_2\right)  _{\rho} &  =H\left(  B_{2}|X_1 X_2\right)  _{\rho
}=\left(  p_{X_1}\left(  0\right)  p_{X_2}\left(  1\right)  +p_{X_1}\left(
1\right)  p_{X_2}\left(  0\right)  \right)  H_{2}\left(  \cos^{2}\left(
\theta\right)  \right)  , \label{eq:entropies-1}\\
H\left(  B_{1}\right)_{\rho}   &  =H_{2}\left(  p_{X_1}\left(  0\right)  +\left(
p_{X_1}\left(  1\right)  p_{X_2}\left(  0\right)  -p_{X_1}\left(  0\right)
p_{X_2}\left(  1\right)  \right)  \sin^{2}\left(  \theta\right)  \right)  ,\\
H\left(  B_{2}\right)_{\rho}   &  =H_{2}\left(  p_{X_2}\left(  0\right)  +\left(
p_{X_1}\left(  0\right)  p_{X_2}\left(  1\right)  -p_{X_1}\left(  1\right)
p_{X_2}\left(  0\right)  \right)  \sin^{2}\left(  \theta\right)  \right)  ,\\
H\left(  B_{2}|X_1\right)_{\rho}   &  =p_{X_1}\left(  0\right)  H_{2}\left(  p_{X_2}\left(
1\right)  \cos^{2}\left(  \theta\right)  \right)  +p_{X_1}\left(  1\right)
H_{2}\left(  p_{X_2}\left(  0\right)  \cos^{2}\left(  \theta\right)  \right)
,\\
H\left(  B_{1}|X_2\right)_{\rho}   &  =p_{X_2}\left(  0\right)  H_{2}\left(  p_{X_1}\left(
1\right)  \cos^{2}\left(  \theta\right)  \right)  +p_{X_2}\left(  1\right)
H_{2}\left(  p_{X_1}\left(  0\right)  \cos^{2}\left(  \theta\right)  \right)  ,\label{eq:entropies-last}
\end{align}
\begin{equation}
\mathbb{E}_{X_{1}^{n},X_{2}^{n},Q^n}\left\{  \frac{1}{M_{1}M_{2}}\sum_{m_{1}%
,m_{2}}\text{Tr}\left\{  \left(  I-\Lambda_{m_{1},m_{2}}\right)  \rho
_{X_{1}^{n}\left(  m_{1}\right)  ,X_{2}^{n}\left(  m_{2}\right)  }^{B_{1}^{n}%
}\right\}  \right\}  \leq\frac{\epsilon}{2}, \label{eq:strong-int-dec-1}
\end{equation}
\begin{equation}
\mathbb{E}_{X_{1}^{n},X_{2}^{n},Q^n}\left\{  \frac{1}{M_{1}M_{2}}\sum_{m_{1}%
,m_{2}}\text{Tr}\left\{  \left(  I-\Gamma_{m_{1},m_{2}}\right)  \rho
_{X_{1}^{n}\left(  m_{1}\right)  ,X_{2}^{n}\left(  m_{2}\right)  }^{B_{2}^{n}%
}\right\}  \right\}  \leq\frac{\epsilon}{2}, \label{eq:strong-int-dec-2}
\end{equation}
\begin{equation}
\mathbb{E}_{X_{1}^{n},X_{2}^{n},Q^n}\left\{  \frac{1}{M_{1}M_{2}}\sum_{m_{1}%
,m_{2}}\text{Tr}\left\{  \left[  \left(  I-\Lambda_{m_{1},m_{2}}^{B_{1}^{n}%
}\right)  +\left(  I-\Gamma_{m_{1},m_{2}}^{B_{2}^{n}}\right)  \right]
\rho_{X_{1}^{n}\left(  m_{1}\right)  ,X_{2}^{n}\left(  m_{2}\right)  }%
^{B_{1}^{n}B_{2}^{n}}\right\}  \right\}  \leq\epsilon. \label{eq:strong-int-both}
\end{equation}
\end{figure*}
where $H_{2}\left(  p\right)  $ is the binary entropy function. We numerically
checked for particular values of $\theta$ whether the conditions
(\ref{eq:VSI-1}-\ref{eq:VSI-2}) hold for all distributions $p_{X_1}\left(
x_1\right)  $ and $p_{X_2}\left(  x_2\right)  $, and we found that they hold when
$\theta\in\left[  0.96,2.18\right]  \cup\left[  4.10,5.32\right]  $ (the latter
interval in the union is approximately a shift of the first interval by $\pi
$). The interval $\left[  0.96,2.18\right]  $ contains $\theta=\pi/2$, the
value of $\theta$ for which the capacity should vanish because the
transformation is equivalent to a full SWAP (the channel at this point has
\textquotedblleft too strong\textquotedblright\ interference). We 
compute the capacity region given in Theorem~\ref{thm:carleial}\ for several
values of $\theta$ in the interval $\theta\in\left[  \pi/2,2.18\right]  $ (it
is redundant to evaluate for other intervals because the capacity region is
symmetric about $\pi/2$ and it is also equivalent for the two $\pi$-shifted
intervals $\left[  0.96,2.18\right]  $ and $\left[  4.1,5.32\right]  $).
Figure~\ref{fig:patrick-example}\ plots these capacity regions for several
values of $\theta$ in the interval $\left[  \pi/2,2.18\right]  $.%
\end{example}

\subsubsection{The Quantum Simultaneous Decoding Inner Bound}

The two-sender quantum simultaneous decoder from Theorem~\ref{thm:two-sender-QSD}
and Corollary~\ref{cor:two-sender-QSD}
allows us to establish a non-trivial
inner bound on the capacity of the quantum interference channel. The strategy is simply
to consider the induced multiple access channels to each receiver and choose the rates low enough
such that each receiver can 
decode the messages from \textit{both} senders \cite{A74,el2010lecture}.
This gives us the following theorem:

\begin{theorem}[Simultaneous Decoding Inner Bound]
\label{thm:sim-inner-bound}Let a ccqq quantum interference channel as in
(\ref{eq:ccqq-int}) be given.
Then an achievable rate region is the union of all rates $R_{1}$ and $R_{2}$ satisfying the below inequalities:
\begin{align*}
R_{1}  &  \leq \min\left\{ I\left(  X_{1};B_{1}|X_{2} Q\right)_{\rho} ,
I\left(  X_{1};B_{2}|X_{2} Q\right)_{\rho} \right\},\\
R_{2}  &  \leq \min\left\{ I\left(  X_{2};B_{2}|X_{1} Q\right)_{\rho} ,
I\left(  X_{2};B_{1}|X_{1} Q\right)_{\rho} \right\},\\
R_1 + R_{2}  &  \leq \min\left\{ I\left(  X_1 X_{2};B_{1}|Q\right)_{\rho} ,
I\left(  X_1 X_{2};B_{2}|Q\right)_{\rho} \right\},
\end{align*}
where the union is over input distributions $p_Q(q) \, p_{X_{1} | Q}\left(x_{1}|q\right)\,
p_{X_{2}|Q}\left(  x_{2}|q\right)$.
\end{theorem}

\begin{proof}
The proof exploits the two-sender quantum simultaneous decoder from Corollary~\ref{cor:two-sender-QSD}.
We first generate a time-sharing sequence $q^n$ according to the product distribution $p_{Q^n}(q^n)$.
Let Sender~1 generate a codebook $\left\{  X_{1}^{n}\left(  m_{1}\right)
\right\}  _{m_{1}}$ independently and randomly according to the distribution
$p_{X_{1}|Q}\left(  x_{1}|q\right)  $, and let Sender~2 generate a codebook
$\left\{  X_{2}^{n}\left(  m_{2}\right)  \right\}  _{m_{2}}$ with the
distribution $p_{X_{2}|Q}\left(  x_{2}|q\right)  $. The induced ccq multiple
access channel to Receiver~1 is $x_{1},x_{2}\rightarrow\rho_{x_{1},x_{2}%
}^{B_{1}}$, and the induced channel to Receiver~2 is $x_{1},x_{2}%
\rightarrow\rho_{x_{1},x_{2}}^{B_{2}}$. Corollary~\ref{cor:two-sender-QSD} states that there exists a
simultaneous decoding POVM\ $\left\{  \Lambda_{m_{1},m_{2}}\right\}  $\ for
Receiver~1 (corresponding to the random choice of code) such that (\ref{eq:strong-int-dec-1})
holds as long as%
\begin{align*}
R_{1}  & \leq I\left(  X_{1};B_{1}|X_{2} Q\right)  _{\rho},\\
R_{2}  & \leq I\left(  X_{2};B_{1}|X_{1} Q\right)  _{\rho},\\
R_{1}+R_{2}  & \leq I\left(  X_{1}X_{2};B_{1} | Q\right)  _{\rho}.
\end{align*}
Similarly, we can invoke Corollary~\ref{cor:two-sender-QSD} to show that there is a simultaneous
decoding POVM $\left\{  \Gamma_{m_{1},m_{2}}\right\}  $ for Receiver~2 such
that (\ref{eq:strong-int-dec-2}) holds
as long as%
\begin{align*}
R_{1}  & \leq I\left(  X_{1};B_{2}|X_{2} Q\right)  _{\rho},\\
R_{2}  & \leq I\left(  X_{2};B_{2}|X_{1} Q\right)  _{\rho},\\
R_{1}+R_{2}  & \leq I\left(  X_{1}X_{2};B_{2} | Q\right)  _{\rho}.
\end{align*}
Thus, if we choose the rates as given in the statement of the theorem, then
all six of the above inequalities are satisfied, implying that the
inequality in (\ref{eq:strong-int-both}) holds.
Invoking the following operator inequality%
\[
I-\Lambda_{m_{1},m_{2}}^{B_{1}^{n}}\otimes\Gamma_{m_{1},m_{2}}^{B_{2}^{n}}\leq
I-\Lambda_{m_{1},m_{2}}^{B_{1}^{n}}+I-\Gamma_{m_{1},m_{2}}^{B_{2}^{n}},
\]
and derandomizing the expectation implies the existence of a code upon which
all parties can agree. The agreed upon code has vanishing error probability in
the asymptotic limit.
\end{proof}

\subsubsection{Strong Interference}

The simultaneous decoding inner bound from the previous section allows us to determine the capacity of a
ccqq interference channel in the setting of \textquotedblleft
strong\textquotedblright\ interference (see page 6-12 of
Ref.~\cite{el2010lecture}). The conditions for \textquotedblleft
strong\textquotedblright\ interference are that the following information
inequalities should hold for all distributions $p_{X_{1}}\left(  x_{1}\right)
$ and $p_{X_{2}}\left(  x_{2}\right)  $:%
\begin{align}
I\left(  X_{1};B_{1}|X_{2}\right)  _{\rho}  &  \leq I\left(  X_{1}%
;B_{2}|X_2\right)  _{\rho},\label{eq:SI-1}\\
I\left(  X_{2};B_{2}|X_{1}\right)  _{\rho}  &  \leq I\left(  X_{2}%
;B_{1}|X_1\right)  _{\rho}, \label{eq:SI-2}%
\end{align}
where $\rho^{X_{1}X_{2}B_{1}B_{2}}$ is a state of the form in (\ref{eq:int-code-state}).

\begin{theorem}[Strong Interference]
\label{thm:strong-in}Let a ccqq quantum interference channel as in
(\ref{eq:ccqq-int}) be given which satisfies the condition of ``strong interference'' as
in (\ref{eq:SI-1}-\ref{eq:SI-2}).
Then the capacity region of such a channel is the union of all rates $R_{1}$ and $R_{2}$ satisfying the below inequalities:
\begin{align*}
R_{1}  &  \leq I\left(  X_{1};B_{1}|X_{2} Q\right)_{\rho} ,\\
R_{2}  &  \leq  I\left(  X_{2};B_{2}|X_{1} Q\right)_{\rho},\\
R_1 + R_{2}  &  \leq \min\left\{ I\left(  X_1 X_{2};B_{1} | Q\right)_{\rho} ,
I\left(  X_1 X_{2};B_{2} | Q\right)_{\rho} \right\},
\end{align*}
where the union is over input distributions $p_Q(q) \, p_{X_{1} | Q}\left(x_{1}|q\right)\,
p_{X_{2}|Q}\left(  x_{2}|q\right)$.
\end{theorem}

\begin{proof} The proof exploits the quantum simultaneous decoding inner bound from
Theorem~\ref{thm:sim-inner-bound} and the strong interference conditions in (\ref{eq:SI-1}-\ref{eq:SI-2}).
The matching outer bound follows from
similar reasoning as on page 6-13 of Ref.~\cite{el2010lecture}, though using quantum information inequalities
rather than classical ones.
\end{proof}

\subsubsection{Han-Kobayashi Achievable Rate Region}

\label{sec:HK}The following result 
provides an achievable rate region for the reliable transmission of classical
data over a \textit{ccqq} quantum interference channel (assuming
Conjecture~\ref{conj:sim-dec}\ regarding the existence of a quantum
simultaneous decoder). We should mention that this result was subsequently proved by Sen \cite{S11a} without relying on Conjecture \ref{conj:sim-dec}. The statement of the theorem generates codes
constructed from a single copy of a ccqq quantum interference channel. We can
obtain the regularization of the region by blocking the channel $k$ times and
constructing codes from the blocked channel (for any finite $k$). 
\begin{theorem}
[Achievable Rate Region for the Quantum Interference Channel]%
\label{thm:han-kobayashi} Assume Conjecture~\ref{conj:sim-dec} holds. Let $\mathcal{S}_{\theta}$ be
the set of tuples of
non-negative reals $\left(  S_{1},S_{2},T_{1},T_{2}\right)  $ such that%
\begin{align}
S_{1}  &  \leq I\left(  U_{1};B_{1}|W_{1}W_{2}\right)  _{\theta}%
,\label{eq:HK-1}\\
T_{1}  &  \leq I\left(  W_{1};B_{1}|U_{1}W_{2}\right)  _{\theta},\\
T_{2}  &  \leq I\left(  W_{2};B_{1}|U_{1}W_{1}\right)  _{\theta},\\
S_{1}+T_{1}  &  \leq I\left(  U_{1}W_{1};B_{1}|W_{2}\right)  _{\theta},\\
S_{1}+T_{2}  &  \leq I\left(  U_{1}W_{2};B_{1}|W_{1}\right)  _{\theta},\\
T_{1}+T_{2}  &  \leq I\left(  W_{1}W_{2};B_{1}|U_{1}\right)  _{\theta},\\
S_{1}+T_{1}+T_{2}  &  \leq I\left(  U_{1}W_{1}W_{2};B_{1}\right)  _{\theta},\label{eq:HK-1st-end}
\end{align}%
\begin{align}
S_{2}  &  \leq I\left(  U_{2};B_{2}|W_{1}W_{2}\right)  _{\theta}, \label{eq:HK-2nd-begin} \\
T_{1}  &  \leq I\left(  W_{1};B_{2}|U_{2}W_{2}\right)  _{\theta},\\
T_{2}  &  \leq I\left(  W_{2};B_{2}|U_{2}W_{1}\right)  _{\theta},\\
S_{2}+T_{1}  &  \leq I\left(  U_{2}W_{1};B_{2}|W_{2}\right)  _{\theta},\\
S_{2}+T_{2}  &  \leq I\left(  U_{2}W_{2};B_{2}|W_{1}\right)  _{\theta},\\
T_{1}+T_{2}  &  \leq I\left(  W_{1}W_{2};B_{2}|U_{2}\right)  _{\theta},\\
S_{2}+T_{1}+T_{2}  &  \leq I\left(  U_{2}W_{1}W_{2};B_{2}\right)  _{\theta},
\label{eq:HK-last}%
\end{align}
where $\theta$ is a state of the following form:%
\begin{multline}
\theta^{U_{1}U_{2}W_{1}W_{2}B_{1}B_{2}}\equiv \\ \sum_{u_{1},u_{2},w_{1},w_{2}%
}p_{U_{1}}\left(  u_{1}\right)  p_{U_{2}}\left(  u_{2}\right)  p_{W_{1}%
}\left(  w_{1}\right)  p_{W_{2}}\left(  w_{2}\right)  \\ \left\vert
u_{1}\right\rangle \left\langle u_{1}\right\vert ^{U_{1}}\otimes\left\vert
u_{2}\right\rangle \left\langle u_{2}\right\vert ^{U_{2}}%
\otimes\left\vert w_{1}\right\rangle \left\langle w_{1}\right\vert ^{W_{1}%
}\otimes\left\vert w_{2}\right\rangle \left\langle w_{2}\right\vert ^{W_{2}%
}\otimes
\label{eq:ipsi-EA-interference-state}\\
\rho_{f_{1}\left(  u_{1},w_{1}\right)  ,f_{2}\left(  u_{2}%
,w_{2}\right)  }^{B_{1}B_{2}},
\end{multline}
and $f_{1}:\mathcal{U}_{1}\times \mathcal{W}_{1}\rightarrow\mathcal{X}_{1}$ and $f_{2}:\mathcal{U}_{2}\times
\mathcal{W}_{2}\rightarrow\mathcal{X}_{2}\mathcal{\ }$are arbitrary functions. A rate
region is achievable if for all $\epsilon>0$ and sufficiently large $n$, there
exists a code with vanishing average error probability as given in (\ref{eq:int-channel-bound})
where $\rho_{f_{1}^{n}\left(  u_{1}^{n}\left(  i\right)  ,w_{1}^{n}\left(
k\right)  \right)  ,f_{2}^{n}\left(  u_{2}^{n}\left(  j\right)  ,w_{2}%
^{n}\left(  m\right)  \right)  }$ represents the encoded state, $i$ is a
\textquotedblleft personal\textquotedblright\ message of Sender~1, $k$ is a
\textquotedblleft common\textquotedblright\ message of Sender~1, $j$ is a
\textquotedblleft personal\textquotedblright\ message of Sender~2, $m$ is a
\textquotedblleft common\textquotedblright\ message of Sender~2, $\left\{
\Lambda_{i,k,m}\right\}  $ is the POVM\ of Receiver~1,
and $\left\{  \Gamma_{j,k,m}\right\}  $ is the POVM\ of
Receiver~2. An achievable rate region for the quantum interference channel
$x_{1},x_{2}\rightarrow\rho_{x_{1},x_{2}}$ is the set of
all rates $\left(  S_{1}+T_{1},S_{2}+T_{2}\right)  $ where $\left(
S_{1},S_{2},T_{1},T_{2}\right)  \in\mathcal{S}_{\theta}$ and $\theta$ is a
state of the form in (\ref{eq:ipsi-EA-interference-state}).
\end{theorem}
\begin{figure*}
\begin{equation}
\frac{1}{L_{1}L_{2}M_{1}M_{2}}\sum_{i,j,k,m}\text{Tr}\left\{  \left(
I-\Lambda_{i,k,m} \otimes\Gamma_{j,k,m}\right)  \rho_{f_{1}^{n}\left(  u_{1}^{n}\left(  i\right)
,w_{1}^{n}\left(  k\right)  \right)  ,f_{2}^{n}\left(  u_{2}^{n}\left(
j\right)  ,w_{2}^{n}\left(  m\right)  \right)  }\right\}  \leq\epsilon,
\label{eq:int-channel-bound}%
\end{equation}
\begin{align}
\mathbb{E}\left\{  \frac{1}{L_{1}M_{1}M_{2}}\sum_{i,k,m}\text{Tr}\left\{
\left(  I-\Lambda_{i,k,m}\right)  \rho_{f_{1}^{n}\left(
u_{1}^{n}\left(  i\right)  ,w_{1}^{n}\left(  k\right)  \right)  ,f_{2}%
^{n}\left(  u_{2}^{n}\left(  j\right)  ,w_{2}^{n}\left(  m\right)  \right)
}\right\}  \right\}   &  \leq\frac{\epsilon}{2},\label{eq:HK-errors-1}\\
\mathbb{E}\left\{  \frac{1}{L_{2}M_{1}M_{2}}\sum_{j,k,m}\text{Tr}\left\{
\left(  I-\Gamma_{j,k,m}\right)  \rho_{f_{1}^{n}\left(
u_{1}^{n}\left(  i\right)  ,w_{1}^{n}\left(  k\right)  \right)  ,f_{2}%
^{n}\left(  u_{2}^{n}\left(  j\right)  ,w_{2}^{n}\left(  m\right)  \right)
}\right\}  \right\}   &  \leq\frac{\epsilon}{2}. \label{eq:HK-errors}
\end{align}
\begin{equation}
\frac{1}{L_{1}L_{2}M_{1}M_{2}}\sum_{i,j,k,m}\text{Tr}\left\{  \left[  \left(
I-\Lambda_{i,k,m}\right)  +\left(  I-\Gamma
_{j,k,m}\right)  \right]  \rho_{f_{1}^{n}\left(
u_{1}^{n}\left(  i\right)  ,w_{1}^{n}\left(  k\right)  \right)  ,f_{2}%
^{n}\left(  u_{2}^{n}\left(  j\right)  ,w_{2}^{n}\left(  m\right)  \right)
}\right\}  \leq\epsilon. \label{eq:HK-bound}
\end{equation}
\end{figure*}

\begin{figure}
[ptb]
\begin{center}
\includegraphics[
natheight=1.932900in,
natwidth=4.540300in,
width=3.5916in
]%
{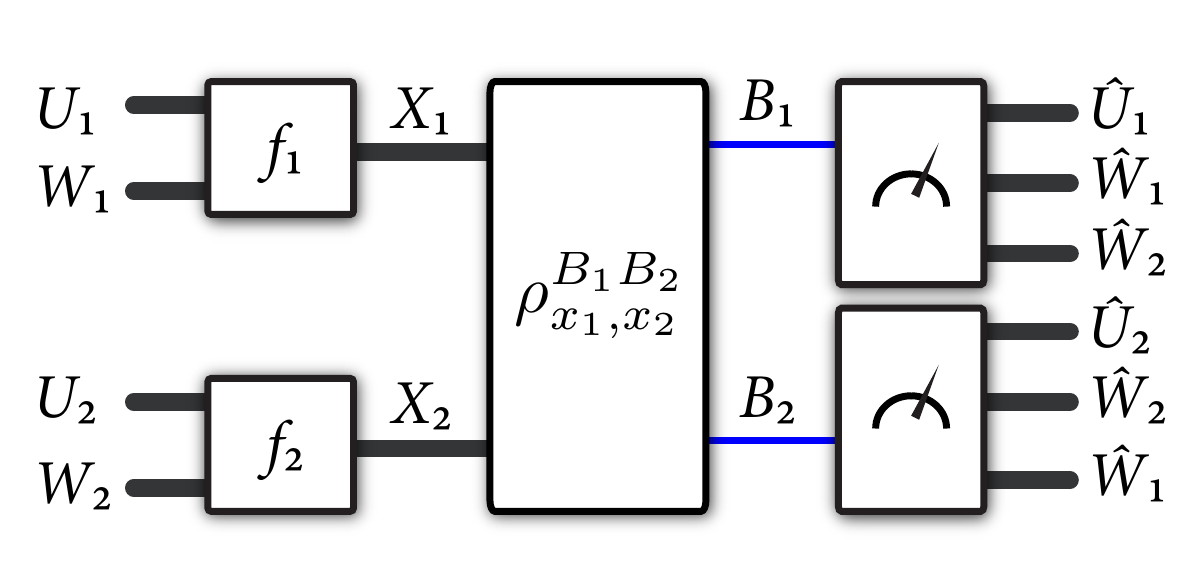}%
\caption{The Han-Kobayashi coding strategy. Sender~1 selects codewords
according to a \textquotedblleft personal\textquotedblright\ random variable
$U_{1}$ and a \textquotedblleft common\textquotedblright\ random variable
$W_{1}$. She then acts on $U_{1}$ and $W_{1}$ with some deterministic function $f_{1}$ that
outputs a variable $X_{1}$ which serves as a classical input to the
interference channel. Sender~2 uses a similar encoding. Receiver~1 performs a
measurement to decode both variables of Sender~1 and the common random
variable $W_{2}$ of Sender~2. Receiver~2 acts similarly. The advantage of
this coding strategy is that it makes use of interference in the channel by
having each receiver partially decode what the other sender is transmitting.
Theorem~\ref{thm:han-kobayashi} gives the rates that are achievable assuming that
Conjecture~\ref{conj:sim-dec}\ holds.}%
\label{fig:han-kob-code}%
\end{center}
\end{figure}
\begin{proof}
We merely need to set up how the senders select a code randomly and the
rest of the proof follows by reasoning similar to that of Han and Kobayashi
\cite{HK81}, although we require an application of
Conjecture~\ref{conj:sim-dec}. Figure~\ref{fig:han-kob-code} depicts the
Han-Kobayashi coding strategy. Sender~1
generates $2^{nS_{1}}$ \textquotedblleft personal\textquotedblright\ codewords
$\left\{  u_{1}^{n}\left(  i\right)  \right\}  _{i\in\left[  1,\ldots
,L_{1}\right]  }$ according to the distribution $p_{U_{1}^{n}}\left(
u_{1}^{n}\right)  $ and $2^{nT_{1}}$ \textquotedblleft
common\textquotedblright\ codewords $\left\{  w_{1}^{n}\left(  k\right)
\right\}  _{k\in\left[  1,\ldots,M_{1}\right]  }$ according to the
distribution $p_{W_{1}^{n}}\left(  w_{1}^{n}\right)  $. Sender~2 generates
$2^{nS_{2}}$ \textquotedblleft personal\textquotedblright\ codewords $\left\{
u_{2}^{n}\left(  j\right)  \right\}  _{j\in\left[  1,\ldots,L_{2}\right]  }$
according to the distribution $p_{U_{2}^{n}}\left(  u_{2}^{n}\right)  $ and
$2^{nT_{2}}$ \textquotedblleft common\textquotedblright\ codewords $\left\{
w_{2}^{n}\left(  m\right)  \right\}  _{m\in\left[  1,\ldots,M_{2}\right]  }$
according to the distribution $p_{W_{2}^{n}}\left(  w_{2}^{n}\right)  $.
Receiver~1 \textquotedblleft sees\textquotedblright\ a three-input multiple
access channel after tracing over Receiver~2's system, and the relevant state
for randomly selecting a code is many copies of Tr$_{B_{2}}\left\{
\theta^{U_{1}U_{2}W_{1}W_{2}B_{1}B_{2}}\right\}  $. Receiver~2
\textquotedblleft sees\textquotedblright\ a three-input multiple access
channel after tracing over Receiver~1's system, and the relevant state for
randomly selecting a code is many copies of Tr$_{B_{1}}\left\{  \theta
^{U_{1}U_{2}W_{1}W_{2}B_{1}B_{2}}\right\}  $. Observe that these states are of
the form needed to apply Conjecture~\ref{conj:sim-dec}. A direct application
of Conjecture~\ref{conj:sim-dec}\ to the state Tr$_{B_{2}}\left\{
\theta^{U_{1}U_{2}W_{1}W_{2}B_{1}B_{2}}\right\}  $ shows that there exists a
POVM\ that can distinguish the common messages of both senders and the
personal message of Sender~1 provided that (\ref{eq:HK-1}-\ref{eq:HK-1st-end}) hold.
Similarly, a direct application of
Conjecture~\ref{conj:sim-dec}\ to the state Tr$_{B_{1}}\left\{  \theta
^{U_{1}U_{2}W_{1}W_{2}B_{1}B_{2}}\right\}  $ shows that there exists a
POVM\ that can distinguish the common messages of both senders and the
personal message of Sender~2 provided
that (\ref{eq:HK-2nd-begin}-\ref{eq:HK-last}) hold. We obtain the bounds
in (\ref{eq:HK-errors-1}-\ref{eq:HK-errors}) on the expectation
of the average error probability for each code, provided that the rates
satisfy the inequalities in (\ref{eq:HK-1}-\ref{eq:HK-last}).
We then sum the two expectations of the average error probabilities together.
Since the expectation is bounded above by some arbitrarily small, positive
number $\epsilon$, there exists a particular code such that the bound in (\ref{eq:HK-bound}) holds. 
We finally apply the bound%
\[
I-\Lambda_{i,k,m}\otimes\Gamma_{j,k,m}
\leq\left(  I-\Lambda_{i,k,m}\right)  +\left(
I-\Gamma_{j,k,m}\right)  ,
\]
that holds for any two commuting positive operators each less than or equal to
the identity, to get the bound in (\ref{eq:int-channel-bound}) on the average
error probability. This demonstrates that any rate pair $\left(  S_{1}+T_{1}%
,S_{2}+T_{2}\right)  $ is achievable for the quantum interference channel (up
to Conjecture~\ref{conj:sim-dec}).
\end{proof}

Extending the strategies of the previous section and this section to the case
of a quantum interference channel with quantum inputs and quantum outputs is
straightforward. The senders have the choice to prepare density operators,
conditional on classical inputs, as input to this general quantum interference
channel, and this extra preprocessing for preparation effectively induces a
\textit{ccqq} quantum interference channel for which they are coding. Thus,
the achievable rate regions include an extra degree of freedom in the choice
of density operators at the inputs. Also, Theorems~\ref{thm:carleial} and
\ref{thm:strong-in} are no
longer optimal in the case of \textquotedblleft very strong\textquotedblright%
\ or ``strong'' interference because entanglement at the individual encoders could increase capacity for
certain interference channels~\cite{H09}.

\subsubsection{Rates achievable by successive decoding}
\label{sec:rate-succ-decoding}
	
	In Section \ref{sec:mac-succ-decoding} on the multiple access channel, we 
	saw that a successive decoding strategy can be used to achieve certain rate 
	tuples. Then, by time-sharing between the different codes achieving these 
	rates, it is possible to construct good codes for the full capacity region of the 
	multiple access channel. 
	To obtain an inner bound for the interference channel, one could try to use these codes
	for the two induced multiple access channels.
	However, this strategy is not well-adapted in this setting
	because the codebooks obtained for the two multiple access channels are not necessarily the same
	for fixed rates $R_1$ and $R_2$. 
	In addition, decoding a codebook constructed by time-sharing between 
	two codebooks $\cC_1$ and $\cC_2$ assumes that both $\cC_1$ and $\cC_2$ 
	are decodable, and these codes
	do in general depend on the properties of the channel for which 
	one is coding. For this reason, a time-sharing strategy that works for one of the induced multiple
	access channels might not work for the other one.

It is however possible to use successive decoding strategies for an interference channel in the following way. We start by considering a strategy where both receivers are asked to decode both messages, i.e., we are dealing with the compound multiple access channel. Such a strategy defines an achievable
rate region known as the ``successive decoding inner bound'' 
for the interference channel (c.f., page 6-7 of Ref.~\cite{el2010lecture}).
	Suppose that Receiver~1 starts 
	by decoding the message of Sender~2 and then the message of Sender~1, 
	and Receiver~2 does the same. 
	We can describe the decode orderings of the receivers by the two permutations
	$\pi_1 = (2,1)$ and  $\pi_2 = (2,1)$.
	In this case, we know that the random code 
	defined by picking $2^{nR_1}$ and $2^{nR_2}$ codewords independently 
	according to the product distributions $p^n_{X^n_1}$ and $p^n_{X^n_2}$ is 
	decodable on average for Receiver~1 provided $R_1 < I(X_1; B_1|X_2)$ 
	and $R_2 < I(X_2; B_1)$. 
	Moreover, it is decodable on average for 
	Receiver~2 provided $R_1 < I(X_2; B_2|X_1)$ and $R_2 < I(X_2; B_2)$. 
	Thus, the rate pairs $R_1 < \min\{I(X_1; B_1|X_2), I(X_1; B_2|X_2\})$ and 
	$R_2 <  \min\{I(X_2; B_1), I(X_2; B_2)\}$ are all achievable for the interference 
	channel. 
	Recall that Receiver~2 is actually not interested in the message 
	sent by Sender~1. The only reason to decode the message of Sender~1 is to 
	be able to decode the message of Sender~2 at a higher rate. It is thus 
	useless to require Receiver~2 to decode the message of Sender~1 after 
	decoding the message of Sender~2. 
	
	The above ordering shows that the rate pairs $R_1, R_2$ where 
	$R_1 < I(X_1; B_1|X_2)$ and $R_2 < \min\{I(X_2; B_1),$ $I(X_2; B_2)\}$ are all 
	achievable for the interference channel. 
	Naturally, we can do the same for all 
	decode orderings $\pi_1$, $\pi_2$ and we can achieve rates arbitrarily close to the following points:
        \begin{align}        	 
		P_1		&= (I(X_1; B_1|X_2), \min\{I(X_2; B_1), I(X_2; B_2) \} ), \label{eq:sd-1}\\
		P_2		&= (\min\{I(X_1; B_1|X_2), I(X_1;B_2)\}, \nonumber\\
		& \,\,\,\,\,\,\,\,\,\,\,\min\{I(X_2;B_1), I(X_2;B_2|X_1)\}), \label{eq:sd-2}\\
		P_3		&= (\min\{I(X_1;B_1), I(X_1;B_2)\}, I(X_2;B_2|X_1)),\label{eq:sd-3} \\
		P_4		&=(I(X_1; B_1), I(X_2; B_2)). \label{eq:sd-4}
        	\end{align}        	
	Of course, one can use time-sharing between these different codes for the interference channel to obtain other achievable rates. These rates are illustrated in the RHS of Figure \ref{fig:jt-succ-decoder}.

	\begin{figure}[ptb]
	\begin{center}
	\includegraphics[width=0.48\textwidth]{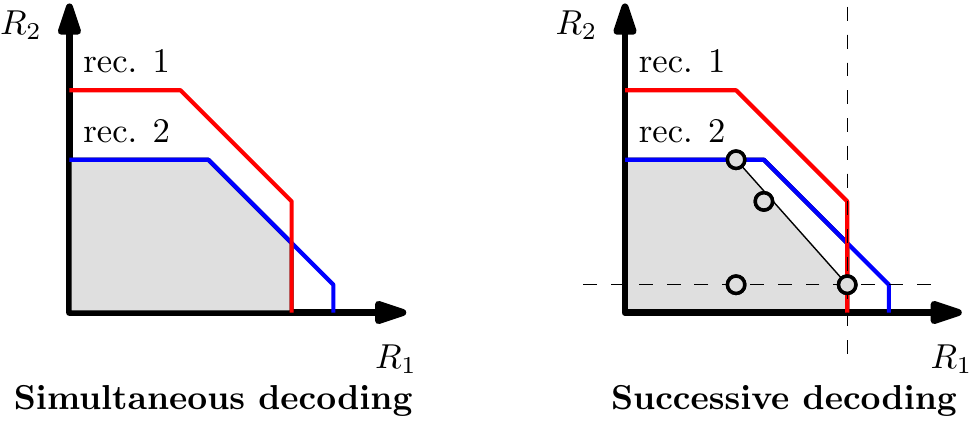}%
	\caption{These plots show achievable rates regions for the interference 
	channel for simultaneous decoding and successive decoding strategies with 
	fixed input distributions. Using a simultaneous decoding strategy, it is 
	possible to achieve the intersection of the two regions of the corresponding 
	multiple access channels. Using a successive decoding strategy, we obtain 
	four achievable rate points that correspond to the possible decoding orders 
	for the two multiple access channels. The solid red and blue lines outline the different 
	multiple access channel achievable rate regions, and the shaded gray areas 
	outline the achievable rate regions for the two different decoding strategies.
	}%
	\label{fig:jt-succ-decoder}%
	\end{center}
	\end{figure}

\textbf{Improving rates using rate-splitting.}
As can be seen in Figure \ref{fig:jt-succ-decoder}, the region defined by the convex hull of  the points \eqref{eq:sd-1}-\eqref{eq:sd-4} is, in general, smaller than the simultaneous decoding inner bound.
A natural question is whether it is possible to obtain the simultaneous decoding inner bound, or even more generally, the full Han-Kobayashi rate region using a more sophisticated successive decoding argument. There exists an attempt to answer this question
for the classical interference channel~\cite{sasoglu2008successive}. This
attempt exploits 
rate-splitting~\cite{GRUW01} and a
careful analysis of the geometrical structure of the four-dimensional region (corresponding to the two natural multiple access channels defined by the interference channel) that projects down to the two-dimensional Chong-Motani-Garg
region~\cite{CMG06}. The Chong-Motani-Garg region is known to be equivalent to the Han-Kobayashi region when
considering all possible input
distributions~\cite{CMGE08,kobayashi2007further}. The argument of Ref.~\cite{sasoglu2008successive} rests on
an assumption that the change of the code distribution dictated by applying
the rate-splitting technique at the convenience of some receiver does not
affect the other receiver's decoding ability.  Unfortunately, this assumption
does not hold in general. We explain this issue in greater detail in the following paragraphs.

Consider an input distribution $p_{X}\left(  x\right)  $ on some alphabet
$\mathcal{X}$. Let $\mathcal{C}_{0}$ be the codebook obtained by picking
$2^{nR}$ independent codewords of length $n$ distributed according to
$p_{X^{n}}\left(  x^{n}\right)  $. A split of $p_{X}\left(  x\right)  $
consists of a function $f:\mathcal{X}\times\mathcal{X}\rightarrow\mathcal{X}$
and distributions $p_{U}\left(  u\right)  $ and $p_{V}\left(  v\right)  $ such
that $f(U,V)\sim p_{X}\left(  x\right)  $ where $U\sim p_{U}\left(  u\right)
$ and $V\sim p_{V}\left(  v\right)  $ are independent~\cite{GRUW01}. The
rate-splitting technique in general refers to following coding strategy.
Generate a code $\mathcal{C}_{U}$ from the distribution $p_{U^{n}}\left(
u^{n}\right)  $ consisting of $2^{nR_{U}}$ independent codewords and a code
$\mathcal{C}_{V}$ from the distribution $p_{V^{n}}\left(  v^{n}\right)  $
consisting of $2^{nR_{V}}$ independent codewords, where $R_{U}+R_{V}=R$. The
codebook $\mathcal{C}_{\text{split}}$ is defined as $\{f^{n}(u^{n}%
,v^{n}):(u^{n},v^{n})\in\mathcal{C}_{U}\times\mathcal{C}_{V}\}$. Note that
$\mathcal{C}_{\text{split}}$ contains $2^{n(R_{U}+R_{V})}=2^{nR}$ codewords.
Furthermore, the codewords of $\mathcal{C}_{\text{split}}$ are all distributed
according to $p_{X^{n}}\left(  x^{n}\right)  $. The difference between this
codebook and $\mathcal{C}_{0}$ is that the codewords in $\mathcal{C}%
_{\text{split}}$ are \textit{not} pairwise independent because two codewords
in $\mathcal{C}_{\text{split}}$ could arise from the same $u^{n}$ and
$v_{1}^{n}\neq v_{2}^{n}$ where $u^{n}\in\mathcal{C}_{U}$ and $v_{1}^{n}%
,v_{2}^{n}\in\mathcal{C}_{V}$.

Now we describe how to choose the rates $R_{U}$ and $R_{V}$. Suppose that
$R=I(X;Y)$ where $Y$ is the output of a channel on input $X$. Then a natural
choice for $R_{U}$ and $R_{V}$ is $R_{U}=I(U;Y)$ and $R_{V}=I(V;Y|U)$ because
$I(X;Y)=I(U;Y)+I(V;Y|U)$. Observe that the values of $R_{U}$ and $R_{V}$
depend on the channel. Consider now a code for an interference channel where
$X$ is to be decoded by both receivers. Such an additional requirement arises
for example for the common messages in the Han-Kobayashi inner bound strategy.
 Let $R=I(X;Y_{1})$ and $R\leq
I(X;Y_{2})$. Using the codebook $\mathcal{C}_{0}$, both receivers are able to
decode $X$. However, when coding for a multiple access channel with output
$Y_{1}$, we might want to split $p_{X}\left(  x\right)  $ into $p_{U}\left(
u\right)  $ and $p_{V}\left(  v\right)  $ and use the codebook $\mathcal{C}%
_{\text{split}}$ for $X$ with rates $R_{U}=I(U;Y_{1})$ and $R_{V}%
=I(V;Y_{1}|U)$ instead of using $\mathcal{C}_{0}$~\cite{GRUW01}. We perform this split
because we want to get a non-corner point of the rate region for the multiple access channel with output $Y_1$
only using successive decoding.
In this case,
Receiver$~1$ can decode with small error probability. We should however keep
in mind that we are coding for an interference channel and we also want
Receiver$~2$ to decode $X$. The problem is that it is possible that
$R_{U}=I(U;Y_{1})>I(U;Y_{2})$, in which case Receiver$~2$ cannot decode $U$
and thus cannot decode $X$. In this case, the code obtained by splitting
according to the first receiver's prescription is not a good code for the second receiver and hence not a good code 
for the
interference channel.

One can however use rate-splitting to obtain potentially better rates than the four points \eqref{eq:sd-1}-\eqref{eq:sd-4} that can be achieved using a simple successive decoding strategy. In fact, splitting the two inputs of the interference channel as in the Han-Kobayashi strategy into a ``personal'' and a ``common'' part and requiring each receiver to decode both common parts induces two 3-user multiple access channels. One can naturally use all $6 \times 6$ pairs of decoding orders to obtain an achievable rate pair for the interference channel. Figure \ref{fig:jt-succ-decoder-rate-splitting} shows some rates that can be achieved using such a strategy for a classical Gaussian interference channel.

	\begin{figure*}[ptb]
	\begin{center}
	\includegraphics[width=6in]{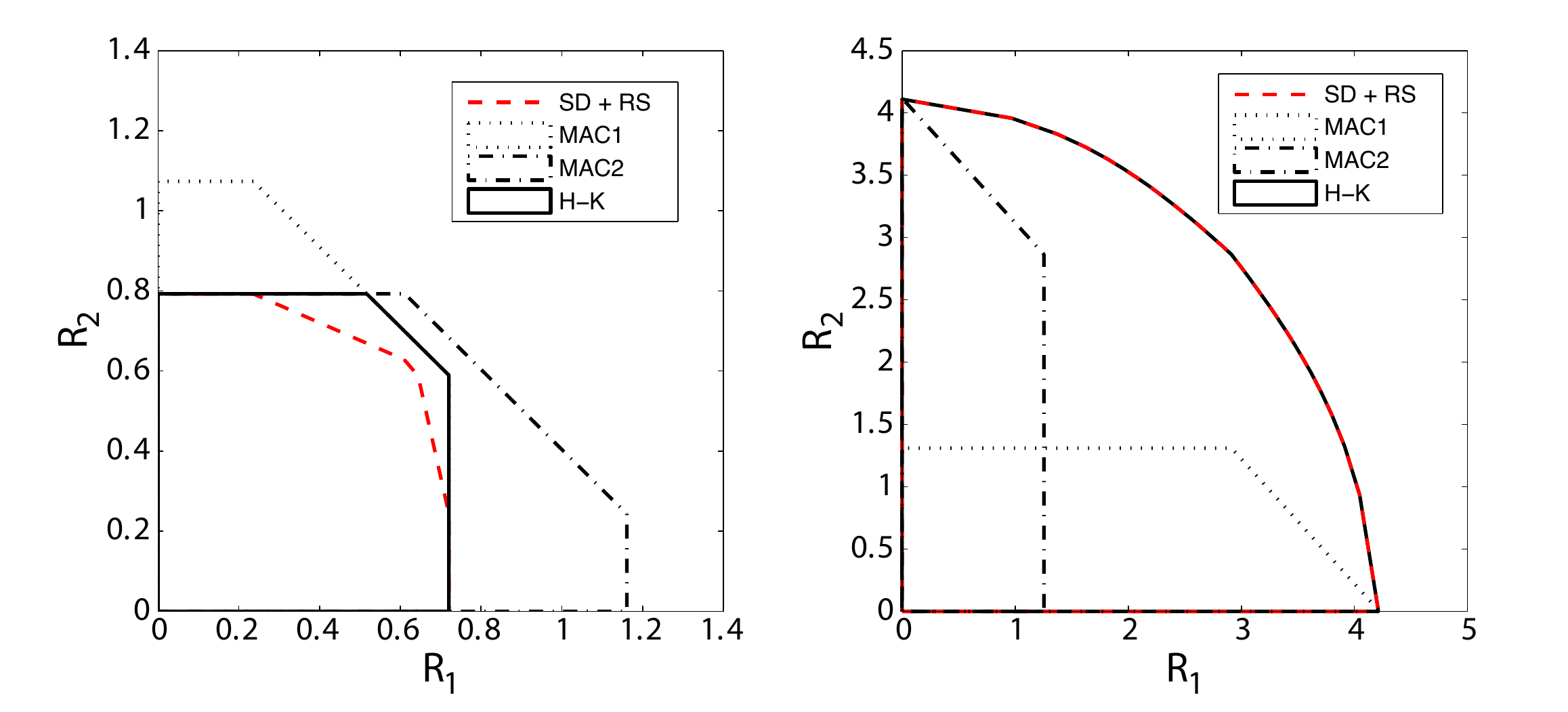}%
	\caption{These two figures plot rate pairs that the senders and receivers in a classical
	Gaussian interference channel can achieve with successive decoding and rate-splitting (SD+RS). The figures
	compare these rates with those achievable by the Han-Kobayashi (HK) coding strategy, while also plotting
	the regions corresponding to the two induced multiple access channels to each
	receiver (MAC1 and MAC2). The LHS figure
	demonstrates that, for a particular choice of signal to noise (SNR)
	 and interference to noise (INR) parameters
	(SNR1 = 1.7, SNR2 = 2, INR1 = 3.4, INR2 = 4),
	successive decoding with rate-splitting does not perform as well as the Han-Kobayashi strategy. The RHS figure
	demonstrates that, for a different choice of parameters 
	(SNR1 =  343, SNR2 = 296, INR1 = 5, INR2 = 5),
	the two strategies perform equally well.}
	\label{fig:jt-succ-decoder-rate-splitting}
	\end{center}
	\end{figure*}

Of course, it is possible to split the inputs even further, leading to two six-user multiple access channels. 
An interesting open question is to determine whether such a strategy can achieve the full Han-Kobayashi region---such a result would be important for the quantum interference channel because it would immediately lead to a way to achieve the analogous Han-Kobayashi region without employing  Conjecture~\ref{conj:sim-dec}.

\subsection{Outer Bound}
\label{sec:outer-bound}

We also give a simple outer bound for the
capacity of the quantum interference channel. This result follows
naturally from a classical result of Sato's \cite{Sato77}, where he observes
that any code for the quantum interference channel also gives codes for three
quantum multiple access channel subproblems, one for Receiver~1, another
for Receiver~2, and a third for the two receivers considered together.
Thus, if we have an outer bound on the underlying quantum
multiple access channel capacities \cite{winter2001capacity}, then we can
trivially get an outer bound on the quantum interference channel capacity. We
omit the following theorem's proof because of its similarity to Sato's proof.

\begin{theorem}
\label{thm:sato-weaker} Consider the Sato region defined as follows:
\begin{equation}
\mathcal{R}_{\text{Sato}}(\mathcal{N})\triangleq 
\bigcup_{p_Q(q)p_{1}(x_{1}|q)p_{2}(x_{2}|q)}\{(R_{1},R_{2})\}, \label{region:Gsato}%
\end{equation}
where $R_{1}$ and $R_{2}$ are rates satisfying the following
inequalities:%
\begin{align}
R_{1}  &  \leq I(X_{1};B_{1}|X_{2}Q)_{\theta},\\
R_{2}  &  \leq I(X_{2};B_{2}|X_{1}Q)_{\theta},\label{Gsato}\\
R_{1}+R_{2}  &  \leq I(X_{1}X_{2};B_{1}B_{2}|Q)_{\theta}.
\end{align}
The above entropic quantities are with respect to the following state%
\begin{multline}
\theta^{QX_{1}X_{2}B_{1}B_{2}}\equiv\sum_{q,x_{1},x_{2}}
p_Q(q)p_{1}(x_{1}|q)p_{2}(x_{2}|q) \ 
|q\rangle\langle q|^{Q}\otimes 
\\|x_{1}\rangle\langle x_{1}|^{X_{1}}\otimes|x_{2}\rangle\langle
x_{2}|^{X_{2}}\otimes\rho_{x_{1}x_{2}}^{B_{1}B_{2}}.
\end{multline}
Then the region $\mathcal{R}_{\text{Sato}}$ forms an outer bound on the
capacity region of the quantum interference channel.
\end{theorem}

\section{The Connection to Unitary Gate Capacities}

\label{sec:gate-capacities}Considerable effort has been devoted to the problem
of establishing the information theoretic capacities of an \emph{interaction}
$U:C\otimes D\rightarrow C\otimes D$ between two quantum
systems~\cite{BennettHLS03,HarrowLeung05,HarrowLeung08,HarrowShor10}. One
imagines that Charlie controls the system represented by the $C$ Hilbert space
while Donna controls $D$, and that they would like to exploit $U$ to
communicate or establish correlations. (More generally, the interaction might
be modeled by a Hamiltonian, but that situation can be reduced to the unitary
case.) Since $U$ has two inputs and two outputs, this is a special case of a
quantum interference channel, and so Theorem~\ref{thm:han-kobayashi} will
yield achievable rates for classical communication over $U$ and, as we shall
see, significantly more.%

\begin{figure}
[ptb]
\begin{center}
\includegraphics[
natheight=1.932900in,
natwidth=4.459800in,
width=3.1117in
]%
{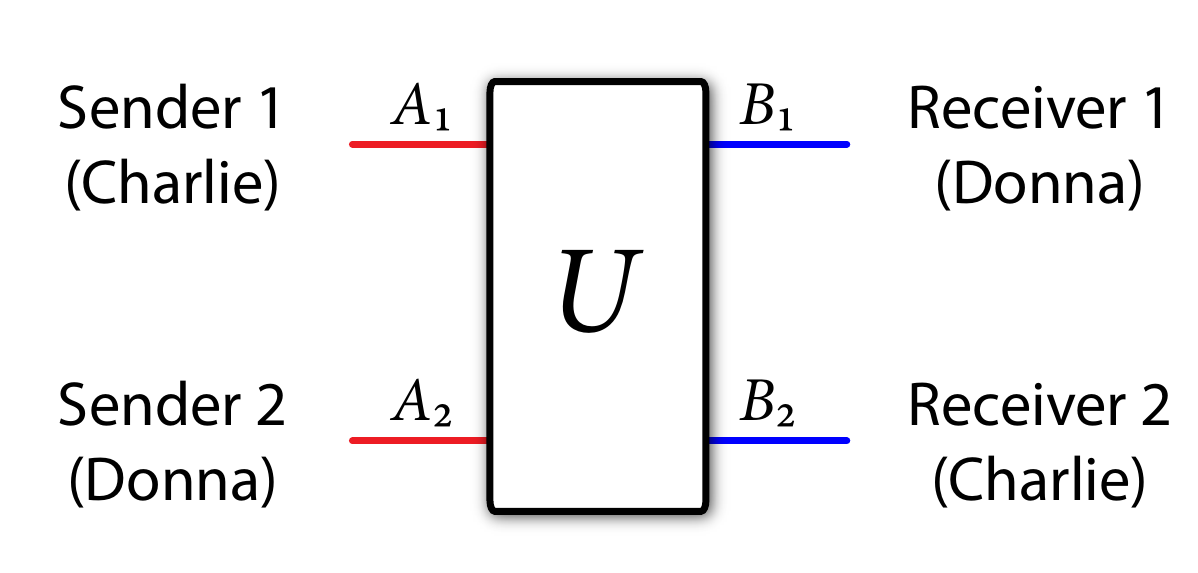}%
\caption{The connection between a quantum interference channel and a
bidirectional unitary gate. The quantum interference channel has quantum
inputs $A_{1}$ and $A_{2}$ and quantum outputs $B_{1}$ and $B_{2}$. We can
identify Sender~1 and Receiver~2 as Charlie and Sender~2 and Receiver~1 as
Donna to make a connection with the bidirectional unitary gate setting.}%
\label{fig:unitary-gate}%
\end{center}
\end{figure}
When $U$ is thought of as an interference channel (say, with quantum inputs
$A_{1}$ and $A_{2}$ and quantum outputs $B_{1}$ and $B_{2}$ as discussed at
the end of Section~\ref{sec:HK}), Charlie plays the roles of both Sender 1 and
Receiver 2, while Donna plays the roles of both Sender 2 and Receiver 1
(Figure~\ref{fig:unitary-gate}\ depicts this communication scenario).
Theorem~\ref{thm:han-kobayashi} then gives achievable rates for simultaneous
Charlie-to-Donna and Donna-to-Charlie classical communication over $U$.
Indeed, it appears to provide the first nontrivial protocol accomplishing this
task for general bidirectional channels. (Earlier protocols assumed free
shared entanglement between Charlie and Donna~\cite{BennettHLS03}.) To apply
the theorem, it suffices to identify $A_{1}=B_{2}=C$ and $A_{2}=B_{1}=D$ in
the interference channel $\mathcal{N}^{A_{1}A_{2}\rightarrow B_{1}B_{2}}%
(\rho)=U\rho U^{\dagger}$. The communication rates achievable
for the $\theta$-SWAP channel of Example~\ref{ex:theta-SWAP}, for instance, apply equally well to this setting.

The fact that Charlie and Donna are each both sender and receiver gives them
some flexibility not available for general interference channels. Most
importantly, in this \textquotedblleft bidirectional\textquotedblright%
\ setting, they are permitted to use $U$ sequentially, reinserting their
outputs into the channel in each successive round~\cite{BennettHLS03}. Codes
for standard interference channels, on the other hand, cannot take advantage
of this flexibility, and so finding the optimal trade-off between forward and
backward communication will likely require codes specifically tailored to the
bidirectional setting.

As an interference channel, $U$ is also special because the \emph{only} noise
is due to interference: the channel itself is noiseless. Because $U$ does not
leak information to an environment, communication can be made coherent at
essentially no cost. This allowed Harrow and Leung to establish the following
remarkable result, which we state informally using resource
inequalities~\cite{DHW08}. Recall that $[c\rightarrow c]$ represents a
classical bit of communication from Charlie to Donna, $[q\rightarrow q]$ one
qubit of communication from Charlie to Donna, and $[q\rightarrow qq]$ one
cobit from Charlie to Donna, that is, the isometry $\sum_{x}\alpha
_{x}|x\rangle^{C}\rightarrow\sum_{x}\alpha_{x}|x\rangle^{C}|x\rangle^{D}$.
$[c\leftarrow c]$~\cite{Harrow04}, $[q\leftarrow q]$ and $[qq\leftarrow q]$
represent the same resources but with Donna the sender and Charlie the
receiver. Finally, $[qq]$ represents a single shared ebit. For a rigorous
definition of resource inequalities, see~\cite{DHW08} and \cite{HarrowShor10}.

\begin{theorem}
[Harrow and Leung~\cite{HarrowLeung05}]For any bipartite unitary (or isometry)
$U$ and $R_{1},R_{2}\geq0$, each of the following resource inequalities is equivalent:%
\begin{align}
\langle U\rangle &  \geq R_{1}[c\rightarrow c]+R_{2}[c\leftarrow c]+E[qq],\\
\langle U\rangle &  \geq R_{1}[q\rightarrow qq]+R_{2}[qq\leftarrow q]+E[qq],\\
\langle U\rangle &  \geq\frac{R_{1}}{2}[q\rightarrow q]+\frac{R_{2}}%
{2}[q\leftarrow q]+\left(  E-\frac{R_{1}+R_{2}}{2}\right)  [qq].
\end{align}

\end{theorem}

Note that the inequalities need only hold in the limit of a large number of uses
of $U$ and might require the catalytic use of resources. Still, they imply
that for bidirectional channels, the codes we have designed for sending
classical data can also be used to send cobits, ebits and even qubits. In
particular, any rates of classical communication that are achievable can
automatically be upgraded to cobit communication rates. While our codes should
be effective for cobit communication, they have not been designed to generate
entanglement. While they can do so at the rate $R_{1} + R_{2}$ by virtue of
the fact that a cobit can be used to generate an ebit, that process might be
inefficient. In fact, Harrow and Leung have even exhibited a particular
channel with $C$ and $D$ each consisting of $k$ qubits for which $R_{1} +
R_{2}$ can never exceed $O(\log k)$ but for which $E$ can be larger than
$k-1$~\cite{HarrowLeung08}. For that channel, our codes would produce an
amount of entanglement exponentially smaller than optimal. Rectifying that
problem would require modifying the interference channel codes we developed in
this article to also establish shared randomness between the two receivers;
such shared randomness would automatically become entanglement in the
bidirectional unitary setting.

\section{Outlook}

Calculating the capacity of the interference channel in the classical setting
has been an open problem for many years now, and calculating the capacity of
the quantum interference channel will be at least as difficult to solve. We
have proved that a quantum simultaneous decoder exists for a multiple access
channel with two senders, and we have given some evidence that it should exist for channels
with three senders. This conjecture holds at least
in the case of a quantum multiple access channel in which certain averages of the channel outputs
commute. If this conjecture holds in the general case, it immediately implies
that the Han-Kobayashi rate region, expressed in terms of Holevo information
quantities, is an achievable rate region for the quantum interference channel. Note that even though the general conjecture is still open, the Han-Kobayashi rate region was recently shown to be achievable \cite{S11a}.

Even though Theorem~\ref{thm:han-kobayashi} is now known to hold \cite{S11a}, it would still be very interesting to prove Conjecture~\ref{conj:sim-dec}. A proof of this conjecture
would probably have important consequences for multiuser quantum
information theory since it would allow for many classical
information theory results based on simultaneous decoding
to be adapted to the quantum setting. It could also likely prove an entanglement-assisted
version of a quantum simultaneous decoder by exploiting the coding techniques
from Ref.~\cite{itit2008hsieh}, and this would in turn lead to another
interesting generalization of the Han-Kobayashi rate region where we assume
that senders share entanglement with their partner receivers. Ref.~\cite{XW11} made progress
in this direction by proving the existence of a quantum simultaneous decoder
for an entanglement-assisted quantum multiple access channel with two senders,
though the three-sender case is still open.


Also, just as there are many different capacities for a single-sender
single-receiver quantum channel, we would expect that there are many
interesting capacities that we could study for a quantum interference channel.
In fact, we initially attempted to use some of the well-known decoupling
techniques for the case of quantum information transmission over the quantum
interference channel~\cite{qcap2008first,arx2006anura}, but we were not able
to achieve non-trivial rates.

Another important question to consider for the \textit{quantum} interference
channel is as follows:\ Is there anything that quantum mechanics can offer to
improve upon the Han-Kobayashi achievable rate region? Quantum effects might
play some unexpected role for the quantum interference channel and allow us to
achieve a rate region that is superior to the well-known Han-Kobayashi rate region.

Finally, it could be that quantum simultaneous decoding is not necessary in
order to achieve the Han-Kobayashi region. In fact, our first attempt at the
proof of Theorem~\ref{thm:han-kobayashi}\ was to quantize the successive
decoding method from Ref.~\cite{sasoglu2008successive}, by exploiting the
coding techniques from Refs.~\cite{winter2001capacity,cmp2005dev} tailored for
classical communication. But we found an issue with the technique in
Ref.~\cite{sasoglu2008successive} even for the classical interference channel
because rate-splitting at the convenience of one receiver affects the other
receiver's decoding abilities. Thus, it remains open to determine if a
successive decoding strategy can achieve the Han-Kobayashi rate region.

We acknowledge discussions with Charlie Bennett, Kamil Br\'{a}dler, Nilanjana
Datta, Fr\'{e}d\'{e}ric Dupuis, Saikat Guha, Aram Harrow, Min-Hsiu Hsieh,
Debbie Leung, Will Matthews, Marco Piani, Eren \c{S}a\c{s}o\u{g}lu, Graeme Smith, John Smolin,
Jon Tyson, Mai Vu, Andreas Winter, and Shen Chen Xu. P.~Hayden
acknowledges support from the Canada Research Chairs program, the Perimeter
Institute, CIFAR, FQRNT's INTRIQ, MITACS, NSERC, ONR through grant
N000140811249, and QuantumWorks. M.~M.~Wilde acknowledges support from the
MDEIE (Qu\'{e}bec) PSR-SIIRI international collaboration grant.
I.~Savov acknowledges support from FQRNT and NSERC. %

\appendix

\section{Typical Sequences and Typical Subspaces}

\label{sec:typ-review}Consider a density operator $\rho$ with the following spectral decomposition:%
\[
\rho=\sum_{x}p_{X}\left(  x\right)  \left\vert x\right\rangle \left\langle
x\right\vert .
\]
The weakly typical subspace is defined as the span of all vectors such that
the sample entropy $\overline{H}\left(  x^{n}\right)  $ of their classical
label is close to the true entropy $H\left(  X\right)  $ of the distribution
$p_{X}\left(  x\right)  $ \cite{book2000mikeandike,W11}:%
\[
T_{\delta}^{X^{n}}\equiv\text{span}\left\{  \left\vert x^{n}\right\rangle
:\left\vert \overline{H}\left(  x^{n}\right)  -H\left(  X\right)  \right\vert
\leq\delta\right\}  ,
\]
where%
\begin{align*}
\overline{H}\left(  x^{n}\right)    & \equiv-\frac{1}{n}\log\left(  p_{X^{n}%
}\left(  x^{n}\right)  \right)  ,\\
H\left(  X\right)    & \equiv-\sum_{x}p_{X}\left(  x\right)  \log p_{X}\left(
x\right)  .
\end{align*}
The projector $\Pi_{\rho,\delta}^{n}$\ onto the typical subspace of $\rho$ is
defined as%
\[
\Pi_{\rho,\delta}^{n}\equiv\sum_{x^{n}\in T_{\delta}^{X^{n}}}\left\vert
x^{n}\right\rangle \left\langle x^{n}\right\vert ,
\]
where we have \textquotedblleft overloaded\textquotedblright\ the symbol
$T_{\delta}^{X^{n}}$ to refer also to the set of $\delta$-typical sequences:%
\[
T_{\delta}^{X^{n}}\equiv\left\{  x^{n}:\left\vert \overline{H}\left(
x^{n}\right)  -H\left(  X\right)  \right\vert \leq\delta\right\}  .
\]
The three important properties of the typical projector are as follows:%
\begin{align*}
\text{Tr}\left\{  \Pi_{\rho,\delta}^{n}\rho^{\otimes n}\right\}    &
\geq1-\epsilon,\\
\text{Tr}\left\{  \Pi_{\rho,\delta}^{n}\right\}    & \leq2^{n\left[  H\left(
X\right)  +\delta\right]  },\\
2^{-n\left[  H\left(  X\right)  +\delta\right]  }\Pi_{\rho,\delta}^{n}  &
\leq\Pi_{\rho,\delta}^{n}\rho^{\otimes n}\Pi_{\rho,\delta}^{n}\leq2^{-n\left[
H\left(  X\right)  -\delta\right]  }\Pi_{\rho,\delta}^{n},
\end{align*}
where the first property holds for arbitrary $\epsilon,\delta>0$ and
sufficiently large $n$.

Consider an ensemble $\left\{  p_{X}\left(  x\right)  ,\rho_{x}\right\}
_{x\in\mathcal{X}}$ of states. Suppose that each state $\rho_{x}$ has the
following spectral decomposition:%
\[
\rho_{x}=\sum_{y}p_{Y|X}\left(  y|x\right)  \left\vert y_{x}\right\rangle
\left\langle y_{x}\right\vert .
\]
Consider a density operator $\rho_{x^{n}}$ which is conditional on a classical
sequence $x^{n}\equiv x_{1}\cdots x_{n}$:%
\[
\rho_{x^{n}}\equiv\rho_{x_{1}}\otimes\cdots\otimes\rho_{x_{n}}.
\]
We define the weak conditionally typical subspace as the span of vectors
(conditional on the sequence $x^{n}$) such that the sample conditional entropy
$\overline{H}\left(  y^{n}|x^{n}\right)  $ of their classical labels is close
to the true conditional entropy $H\left(  Y|X\right)  $ of the distribution
$p_{Y|X}\left(  y|x\right)  p_{X}\left(  x\right)  $ \cite{book2000mikeandike,W11}:%
\[
T_{\delta}^{Y^{n}|x^{n}}\equiv\text{span}\left\{  \left\vert y_{x^{n}}%
^{n}\right\rangle :\left\vert \overline{H}\left(  y^{n}|x^{n}\right)
-H\left(  Y|X\right)  \right\vert \leq\delta\right\}  ,
\]
where%
\begin{align*}
\overline{H}\left(  y^{n}|x^{n}\right)    & \equiv-\frac{1}{n}\log\left(
p_{Y^{n}|X^{n}}\left(  y^{n}|x^{n}\right)  \right)  ,\\
H\left(  Y|X\right)    & \equiv-\sum_{x}p_{X}\left(  x\right)  \sum_{y}%
p_{Y|X}\left(  y|x\right)  \log p_{Y|X}\left(  y|x\right)  .
\end{align*}
The projector $\Pi_{\rho_{x^{n}},\delta}$ onto the weak conditionally typical
subspace of $\rho_{x^{n}}$ is as follows:%
\[
\Pi_{\rho_{x^{n}},\delta}\equiv\sum_{y^{n}\in T_{\delta}^{Y^{n}|x^{n}}%
}\left\vert y_{x^{n}}^{n}\right\rangle \left\langle y_{x^{n}}^{n}\right\vert ,
\]
where we have again overloaded the symbol $T_{\delta}^{Y^{n}|x^{n}}$ to refer
to the set of weak conditionally typical sequences:%
\[
T_{\delta}^{Y^{n}|x^{n}}\equiv\left\{  y^{n}:\left\vert \overline{H}\left(
y^{n}|x^{n}\right)  -H\left(  Y|X\right)  \right\vert \leq\delta\right\}  .
\]
The three important properties of the weak conditionally typical projector are
as follows:%
\begin{align*}
\mathbb{E}_{X^{n}}\left\{  \text{Tr}\left\{  \Pi_{\rho_{X^{n}},\delta}%
\rho_{X^{n}}\right\}  \right\}    & \geq1-\epsilon,\\
\text{Tr}\left\{  \Pi_{\rho_{x^{n}},\delta}\right\}    & \leq2^{n\left[
H\left(  Y|X\right)  +\delta\right]  },\\
2^{-n\left[  H\left(  Y|X\right)  +\delta\right]  }\ \Pi_{\rho_{x^{n}}%
,\delta}  & \leq\Pi_{\rho_{x^{n}},\delta}\ \rho_{x^{n}}\ \Pi_{\rho_{x^{n}%
},\delta} \\ & \leq2^{-n\left[  H\left(  Y|X\right)  -\delta\right]  }\ \Pi
_{\rho_{x^{n}},\delta},
\end{align*}
where the first property holds for arbitrary $\epsilon,\delta>0$ and
sufficiently large $n$, and the expectation is with respect to the
distribution $p_{X^{n}}\left(  x^{n}\right)  $.

\section{Gentle Operator Lemma}

\label{sec:useful-lemmas}

\begin{lemma}[Gentle Operator Lemma for Ensembles \cite{itit1999winter,ON07,W11}]\label{lem:gentle-operator}
Given an ensemble $\left\{  p_{X}\left(  x\right)  ,\rho_{x}\right\}  $ with
expected density operator $\rho\equiv\sum_{x}p_{X}\left(  x\right)  \rho_{x}$,
suppose that an operator $\Lambda$ such that $I\geq\Lambda\geq0$ succeeds
with high probability on the state $\rho$:%
\[
\text{Tr}\left\{  \Lambda\rho\right\}  \geq1-\epsilon.
\]
Then the subnormalized state $\sqrt{\Lambda}\rho_{x}\sqrt{\Lambda}$ is close
in expected trace distance to the original state $\rho_{x}$:%
\[
\mathbb{E}_{X}\left\{  \left\Vert \sqrt{\Lambda}\rho_{X}\sqrt{\Lambda}%
-\rho_{X}\right\Vert _{1}\right\}  \leq2\sqrt{\epsilon}.
\]

\end{lemma}

\bibliographystyle{plain}
\bibliography{interferenceChannel}

\end{document}